\documentclass[11pt]{article}
\usepackage{amsmath}
\usepackage{amsfonts}
\usepackage{amsthm}
\newtheorem{thm}{Theorem}[section]
\newtheorem{lem}{Lemma}[section]

\theoremstyle{remark}
\newtheorem{rem}{Remark}
\numberwithin{equation}{section}
\newcommand \RR   {\mathbb{R}}
\def\be{\begin{equation}}
\def\ee{\end{equation}}

\def\b{\bullet}
\def\la{\label}

\def\pa{\partial}
\def\f{\frac}

\def\t{\tilde}

\def\g{\gamma}

\def\r{\rho}
\def\b{\bar}
\def\f{\frac}

\def\t{\tau}

\begin{document}
\title{ Nonlinear Dynamical Stability of Newtonian Rotating  White Dwarfs and Supermassive Stars }
\author{ Tao Luo \&  Joel Smoller}
\date{}
\maketitle
\begin{abstract}We prove  general nonlinear stability and existence  theorems for  rotating star solutions which are axi-symmetric steady-state solutions of the compressible isentropic Euler-Poisson
equations in 3 spatial dimensions. We apply our results to rotating and non-rotating white dwarf,  and rotating high density supermassive (extreme relativistic) stars, stars which are in convective equilibrium and have uniform chemical composition.  This paper is a continuation of our earlier  work (\cite{LS1}).\end{abstract}
\tableofcontents

\section{Introduction}

The motion of a
 compressible isentropic perfect fluid with  self-gravitation is
modeled by the  Euler-Poisson equations in three space dimensions  (cf \cite{ch}):
\begin{equation}\label{1.1}
\begin{cases}
&\rho_t+\nabla\cdot(\rho {\bf v})=0,\\
&(\rho {{\bf v}})_t+\nabla\cdot(\rho {\bf v}\otimes {\bf v})+\nabla p(\rho)=-\rho \nabla \Phi,\\
&\Delta \Phi=4\pi \rho.
\end{cases}
\end{equation}
Here $\rho$, ${\bf v}=(v_1, v_2, v_3)$, $p(\rho)$ and $\Phi$ denote
the density, velocity, pressure and gravitational
potential, respectively.
  The gravitational potential is given by
\be\label{phi}\Phi(x)=-\int_{\RR^3} \frac{\rho(y)}{|x-y|}dy =-\rho\ast \frac{1}{|x|},\ee where $\ast$ denotes convolution.  System (\ref{1.1}) is used to model the evolution of a Newtonian gaseous star (\cite {ch}). In the  study of time-independent solutions of system (\ref{1.1}), there are two cases, non-rotating stars and rotating stars. An important question concerns the stability of such solutions.
Physicists call such star solutions stable provided that they are minima of an associated energy functional (\cite{wein}, p.305 \& \cite{ST}). Mathematicians, on the other hand, consider dynamical nonlinear stability via solutions of  the Cauchy problem. The main purpose of this paper is to prove a general theorem which relates these two notions and shows that for a wide class of Newtonian rotating stars,  minima of the energy functional are in fact, {\it dynamically} stable. This is done for
various equations of state $p=p(\r)$ which includes polytropes, supermassive,  and white dwarf stars.

For non-rotating stars, Rein (\cite{Rein}) has proved  nonlinear stability under  various hypotheses on the equation of state, including
in particular,  polytropes where $p=k\r^{\g}$, $\g>4/3$;  his theory applies to neither white dwarf nor supermassive stars. In a recent paper,
\cite{LS1}, we studied nonlinear stability of {\it rotating} polytropic stars, where $p=k\r^{\g}$, $\g>4/3$. In this paper, we generalize these results to rotating white dwarf and supermassive stars,  thereby completing the nonlinear stability theory for rotating (and non-rotating) compressible Newtonian stars*.\begin{figure}[b]\rule[-2.5truemm]{5cm}{0.1 truemm}
{\footnotesize \\
$^*$   In all cases under consideration,  stability is only
 ``conditional'' because no global in time solutions have been constructed so far for compressible Euler-type equations in three spatial dimensions;
 this is a major open problem.}\end{figure}

Our main theorem applies to minimizers of an energy functional with a total mass constraint. The crucial hypotheses are that
the infimum of the energy functional in the requisite class,  be finite and  negative. This is verified for both white dwarf and supermassive stars
by combining a scaling technique used by  Rein (\cite{rein1}), together with our method in \cite{LS1} where we use some particular solutions
of the Euler-Poisson equations in order to simplify the energy functional.  It should be noticed that neither the scaling technique in \cite{rein1} nor the method in \cite{LS1} using particular solutions of Euler-Poisson equations apply to white dwarf stars directly.  As a bi-product of our method, we prove the existence of a minimizer
for the energy functional, which is a rotating white dwarf star solution, in a class of functions having less symmetry than those solutions obtained in
\cite{AB} and \cite{FT2}. The method in \cite{AB} and  \cite{FT2} is to construct a specific minimizing sequence of the energy functional, each element in the sequence being a steady solution of the Euler-Poisson equations. In contrast, our method is to show that {\it any} minimizing sequence of the energy functional must be compact (cf. Theorem 3.1 below). This fact is crucial for both existence and stability results.

For a white dwarf star (a star in which gravity is balanced by electron degeneracy pressure), the pressure function $p(\r)$ obeys the following asymptotics (\cite{ch}, Chapter 10):
\be\label{wd1} \begin{cases} & p(\r)=c_1\r^{4/3}-c_2\r^{2/3}+\cdots, \qquad \r\to\infty,\\
& p(\r)=d_1\r^{5/3}-d_2\r^{7/3}+O(\r^3), \qquad \r\to 0,\end{cases}\ee
where  $c_1$, $c_2$, $d_1$ and $d_2$ are positive constants. The existence theory for non-rotating white dwarf stars is classical provided the mass M of the star is not greater than a critical mass $M_c$ ($ M\le M_c$) (\cite{ch}). For rotating white dwarf stars with prescribed total
mass and angular momentum distribution, Auchumuty and Beals (\cite{AB}) proved that if the angular momentum distribution is nonnegative, then
existence holds if $M\le M_c$. Friedman and Turkington (\cite{FT2}) proved  existence for any mass provided that the angular momentum distribution is everywhere positive; see Li (\cite{Li1}),  Chanillo \& Li
(\cite{Li2}) and Luo \& Smoller (\cite{LS}) for related results for  rotating star solutions with prescribed constant angular velocity. To the best of our knowledge, our stability theorem in this paper for rotating and non-rotating white dwarf stars with $M\le M_c$ is the first nonlinear dynamical stability theorem for such stars.

 For a supermassive star (a star  which is supported by the pressure of radiation
rather than that of matter; sometimes called an extreme relativistic degenerate star \cite{ST}), the pressure $p(\r)$ is given by (\cite{wein}):
\be\label{supermassive} p(\r)=k \r^{\gamma},\  \gamma=4/3,\ee
where $k>0$ is a constant. For non-rotating spherically symmetric solutions for  supermassive stars, Weinberg (\cite{wein}) showed
that the total energy vanishes; thus to quote Weinberg (\cite{wein}, p. 327) ``the polytrope with $\gamma=4/3$ is trembling between stability
and instability'', and he remarks that one  needs to use general relativity to settle this stability problem. For rotating supermassive star solutions, we show here that the energy is negative $E<0$ due to the rotational kinetic energy (see (4.26) below). Thus the stability problem falls within the framework of Newtonian mechanics and so our general stability theorem applies to show that {\it rotating} supermassive stars are nonlinearly stable, provided that $M\le M_c$.

For the stability of both white dwarfs and supermassive stars, we require that the total mass of each one lies below to a corresponding critical mass,
a ``Chandrasekhar'' limit.  We show that this holds because
the pressure function for both is of the order $\r^{4/3}$ as $\r\to\infty$.

The above dynamical stability results for rotating stars apply for  axi-symmetric perturbations.  For   non-rotating
stars,  G. Rein (\cite{Rein}) proved nonlinear dynamical stability for  general perturbations. However, his result does not apply to
white dwarf stars. For non-rotating white dwarf stars, the problem was formulated by Chandrasekhar \cite{chan} in 1931 (and also in \cite{fowler} and \cite{landau}) and leads to an equation for the density which was called the `` Chandrasekhar equation '' by Lieb and Yau in \cite{liebyau}. This equation predicts the gravitational collapse at some critical mass (\cite{chan} and \cite{ch}). This gravitational collapse was also verified by Lieb and Yau (\cite{liebyau}) as the limit of Quantum Mechanics. In Section 5,  we  prove the nonlinear dynamical stability for non-rotating white dwarf stars with general perturbations provided that  the total mass is below  some critical mass.

Other related results besides those mentioned above for compressible fluid rotating stars can be found in \cite{Au},  \cite{CF}, \cite{FT1},   and  \cite{LS}.

The linearized stability and instability  for non-rotating  and rotating stars were   discussed by Lin (\cite{lin} ), Lebovitz (\cite{lebovitz})  and Lebovitz \& Lifschitz (\cite{lebovitz1}). Related nonlinear stability and instability results for gaseous stellar objects can be found in Guo \& Rein (\cite{guo1},
 \cite{guo2}) and Jang (\cite{Jang}). Related results for the Euler-Poisson equations of self-gravitating fluids
can be found in \cite{LY}, \cite{Humi}, \cite{MK} and \cite{Wang}.

\section{ Rotating Star Solutions}
We now introduce some notation which will be used throughout this paper. We use $\int$ to denote $\int_{\RR^3}$, and use $||\cdot||_q$ to denote $||\cdot||_{L^q(\RR^3)}$.  For any point $x=(x_1, x_2, x_3)\in \RR^3$, let
\be\label{1.5'} r(x)=\sqrt{x_1^2+x_2^2}, \  z(x)=x_3,\  B_R(x)=\{y\in \RR^3, \ |y-x|<R\}.\ee
For any function $f\in L^{1}(\RR^3)$, we define the operator $B$ by
\be\label{B} B f(x)=\int \frac{f(y)}{|x-y|}dy =f\ast \frac{1}{|x|}.\ee
Also, we use $\nabla$ to denote the spatial gradient, i.e.,  $\nabla=\nabla_x=(\pa_{x_1},\ \pa_{x_2}, \ \pa_{x_3})$.  $C$ will denote a generic positive constant.

A rotating star solution $(\tilde \rho, \tilde {\bf v},\tilde \Phi)(r, z)$, where $r=\sqrt{x_1^2+x_2^2}$  and $z=x_3$, $x=(x_1, x_2, x_3)\in \RR^3,$  is an {\it axi-symmetric}  time-independent solution of system (\ref{1.1}), which models a star rotating about the $x_3$-axis. Suppose the angular momentum (per unit mass),  $J(m_{\tilde \r}(r))$ is prescribed, where
\be m_{\tilde \r}(r)=\int_{\sqrt{x_1^2+x_2^2}<r}\tilde \r(x)dx=\int_0^r 2\pi s\int_{-\infty}^{+\infty}\tilde \r(s, z)dsdz, \ee
is the mass in the cylinder $\{x=(x_1, x_2, x_3): \sqrt{x_1^2+x_2^2}<r\}$, and $J$ is a given function.
   In this case, the velocity field
   $\tilde {\bf v}(x)=(v_1, v_2, v_3)$ takes the form $$\tilde {\bf v}(x)=(-\f{x_2 J(m_{\tilde \r}(r))}{r^2},  \f{x_1 J(m_{\tilde \r}(r))}{r^2}, 0). $$ Substituting  this in (\ref{1.1}), we find that $\tilde \r(r, z)$ satisfies the following two equations:
   \be\label{03}\begin{cases} &\partial_r p(\tilde \r)=\tilde \r\pa_r (B \tilde \r)+\tilde \r L(m_{\tilde \r}(r)r^{-3}, \\
                     &\partial_z p(\tilde \r)=\tilde \r\pa_z (B \tilde \r),
\end{cases}\ee
where the operator $B$ is defined in (\ref{B}), and
 $$ L(m_{\tilde \r})=J^2(m_{\tilde \r})$$ is the square of the angular momentum.
We define
\be\label{z1}
A(\r)=\r\int_0^{\r}\f{p(s)}{s^2}ds. \ee
  It is easy to verify that (cf. \cite{AB}) (\ref{03}) is equivalent to
  \be\label{z2}
  A'(\tilde \r(x))+\int_{r(x)}^{\infty}L(m_{\tilde \r}(s)s^{-3}ds-B\tilde \r(x)=\lambda, \qquad {\rm where~} \tilde \r(x)>0,\ee
  for some constant $\lambda$.   Here $r(x)$ and  $z(x)$ are as in (\ref{1.5'}).
  Let $M$ be a positive constant and let $W_M$ be the set of functions $\r$ defined by ,
\begin{align*}W_M=&\{\r: \RR^3\to \RR,\ \r  {\rm~is~axisymmetric, ~}\rho\ge 0, a.e., \\
 &\int\rho(x)dx=M, \ \int\left(A(\r(x))+\f{\rho(x)L(m_{\r}(r(x)))}{r(x)^2}+\r(x)B\r(x)\right)dx<+\infty.\}\end{align*}
 For $\rho\in W_M$,  we define the {\bf energy functional} $F$ by
  \begin{align}\label{E}
  F(\rho)&=\int
  [A(\rho(x))+\f{1}{2}\f{\rho(x)L(m_{\r}(r(x)))}{r(x)^2}-\frac{1}{2}\rho(x) B\rho(x)]dx.
  \end{align}
  In (\ref{E}),  the first
 term denotes the potential energy, the middle term denotes the rotational kinetic energy and the third term is the gravitational energy.

 For a white dwarf star, the pressure function $p(\r)$ satisfies the following conditions:
 \be\label{whitedwarf}
 \lim_{\r\to 0+}\frac{p(\r)}{\r^{4/3}}=0,\  \lim_{\r\to \infty}\frac{p(\r)}{\r^{4/3}}=\mathfrak{K},\  p'(\r)>0 {~\rm as~} \r>0,
 \ee
 where $\mathfrak{K}$ is a finite positive constant.
  Assuming  that the function  $L\in C^1[0, M]$ and satisfies
  \be\label{L} \ L(0)=0,\ L(m)\ge 0, \ for~ 0\le m\le M,\ee
 Auchmuty and Beals (cf. \cite{AB}) proved the existence of a minimizer of the
functional $F(\rho)$ in the class of functions  $W_{M, S}=W_M\cap W_S$, where
  \begin{equation}\label{W2}
  W_S=\{ \r: \RR^3\to \RR,\   \r(x_1, x_2, -x_3)=\r(x_1, x_2, x_3),\ x_i\in \RR, i=1,\ 2, \ 3\}.
  \end{equation} Their result is given  in the following theorem.
\begin{thm}\label{aa1}(\cite{AB}). If the pressure function $p$ satisfies (\ref{whitedwarf})\  (for either $0<\mathfrak{K}<+\infty$ or $\mathfrak{K}=+\infty$  and (\ref{L}) holds, then there exists a constant
$M_c>0$ depending on the constant $\mathfrak{K}$ in (\ref{whitedwarf})(if $\mathfrak{K}=+\infty$ then $M_c=+\infty$, if $0<\mathfrak{K}<+\infty$, then $0<M_c<+\infty$) such that, if
\be\label{criticalmass} M<M_c,
\ee
  then there exists a function $\hat \r(x)\in W_{M, S}$ which minimizes $F(\rho)$ in  $W_{M, S}$. Moreover, if
\be\label{G}
G=\{x\in \RR^3:\  \hat \r(x)>0\},\ee
then $\bar G$ is a compact set in $\RR^3$, and $\hat \r\in C^1(G)\cap C^{\beta}(\RR^3)$ for some $0<\beta<1$. Furthermore, there exists a constant $\mu<0$
such that
\be\label{lambda}
\begin{cases}
& A'(\hat \r(x))+\int_{r(x)}^{\infty}L(m_{\hat \r}(s)s^{-3}ds-B\hat \r(x)=\mu, \qquad x\in G,\\
&\int_{r(x)}^{\infty}L(m_{\hat \r}(s)s^{-3}ds-B\hat \r(x)\ge \mu, \qquad x\in \RR^3-G.\end{cases}\ee
\end{thm}
\begin{rem} When $0<\mathfrak{K}<\infty$, the constant $0<M_c<+\infty$ in (\ref{criticalmass}) is called critical mass. The critical mass was first found by Chandrasekhar (cf.\cite{ch}) in the study of non-rotating white dwarf stars.
When $0<\mathfrak{K}<\infty$, it was proved by  Friedman and  Turkington (\cite{FT2}) that,  if the angular momentum satisfies the following
condition
\be J\in C^1([0, M]),\  J'(m)\ge 0, {~\rm for~} 0\le m\le M,
 J(0)=0,\  J(m)>0 {~\rm for~} 0<m\le M,\ee
where $J$ is the angular momentum, then the condition (\ref{criticalmass}) can be removed, i.e., the above theorem holds for
any positive total mass $M$.\end{rem}

In this paper, we are interested in the  minimizer of functional $F$ in the {\it larger} class  $W_M$. By the same argument as in \cite{AB}, it is easy to prove the following theorem on the regularity of the minimizer.
\begin{thm}\label{ab} Suppose that the pressure function $p$ satisfies:
\be\label{p1} \lim_{\r\to 0+}\frac{p(\r)}{\r^{6/5}}=0,\  \lim_{\r\to \infty}\frac{p(\r)}{\r^{6/5}}=\infty,\  p'(\r)>0 {~\rm as~} \r>0,
 \ee
 and the angular momentum satisfies (\ref{L}).
Let $\tilde \r$ be a minimizer of the energy functional $F$ in $W_M$ and let \be\label{G1}
\Gamma=\{x\in \RR^3:\  \tilde \r(x)>0\}, \ee then $\tilde \r\in C(\RR^3)\cap C^1(\Gamma)$.
Moreover,
 there exists a constant $\lambda$
such that
\be\label{lambda11}
\begin{cases}
& A'(\tilde \r(x))+\int_{r(x)}^{\infty}L(m_{\tilde \r}(s)s^{-3}ds-B\tilde \r(x)=\lambda, \qquad x\in \Gamma,\\
&\int_{r(x)}^{\infty}L(m_{\tilde \r}(s)s^{-3}ds-B\tilde \r(x)\ge \lambda, \qquad x\in \RR^3-\Gamma.\end{cases}\ee
\end{thm}

\noindent We call such a minimizer $\tilde \r$ a {\it rotating star} solution with  total mass $M$ and angular momentum
$\sqrt{ L(m)}$.

\section{General Existence and Stability Theorems}
For the angular momentum, besides the condition (\ref{L}), we also assume that it satisfies the following conditions:
\be\label{L1} L(a m)\ge a^{4/3}L(m), \ 0<a\le 1,\ 0\le m\le M,\ee
\be\la{L2'} L'(m)\ge 0,\qquad 0\le m\le M.\ee
Condition (\ref{L2'}) is called the S$\ddot{\rm o}$lberg stability criterion (\cite{Ta}).
\subsection{Compactness of Minimizing Sequence}
In this section, we first establish a compactness result for the minimizing sequences of the functional $F$ . This compactness
result is crucial for the existence and stability analyses.
\begin{thm}\label{aa}  Suppose  that the square of the angular momentum  $L$ satisfies (\ref{L}), (\ref{L1}) and (\ref{L2'}), and the pressure function $p$ satisfies the following conditions
\be\label{whitedwarf1}
 p\in C^1[0, +\infty), \int_0^1\f{p(\r)}{\r^2}d\r<+\infty,  \lim_{\r\to \infty}\frac{p(\r)}{\r^{\gamma}}=K,\ p(\r)\ge 0,  p'(\r)>0  {~\rm for~} \r>0,
 \ee
 where   $0<K<+\infty$  and $\gamma\ge 4/3$.  If\\
 (1) \be\label{negative}  \inf_{\r\in W_M}F(\r)<0,\ee
 and\\
 (2) for $\r\in W_M$, \be\label{us} \int [A(\r)(x)+\frac{1}{2}\f{\rho(x)L(m_{\r}(r(x)))}{r(x)^2}] dx\le C_1 F(\r)+C_2,\ee
for some positive constants $C_1$ and $C_2$,
then the following hold:\\
(a)  If  $\{\r^i\}\subset W_M $ is a minimizing sequence for the functional $F$,
then there exist a sequence of vertical shifts $a_i{\bf e_3}$ ($a_i\in \RR$, ${\bf e_3}=(0, 0, 1)$),  a subsequence of $\{\r^i\}$,  (still labeled  $\{\r^i\}$),  and a function $\tilde \r\in W_M$, such that for any $\epsilon>0$ there exists $R>0$ with
\be\label{2.15}
\int_{|x|\ge R}T\r^i(x)dx\le \epsilon, \quad i\in \mathbb{N},\ee
and
\be\label{2.16} T\r^i(x)\rightharpoonup \tilde \r,\  weakly~in~L^{\gamma}(\RR^3),\ as\  i\to \infty,\ee
where $T\r^i(x):=\r^i(x+a_i{\bf e_3})$.\\
\noindent Moreover\\
(b)
\be\label{2.17}\nabla B (T\rho^i)\to \nabla B(\tilde \r)~ strongly~ in~L^2(\RR^3),\ as\  i\to \infty. \ee
 (c) $\tilde \r$ is a minimizer of $F$ in $W_M$.
\end{thm}

Thus $\tilde \r$ is a rotating star solution with  total mass $M$ and angular momentum
$\sqrt {L}$.

\begin{rem} i) The assumption (\ref{negative}) is crucial for our compactness and stability analysis. The physical meaning of this is that
the gravitational energy, the negative part of the energy $F$, should be greater than the positive part, which means the gravitation
   should be strong enough to hold the star together. In section 4, we will verify this assumption. Roughly speaking, in addition to (\ref{whitedwarf1}),
   if we require
   \be\label{origin} \lim_{\r\to 0+}\frac{p(\rho)}{\r^{\gamma_1}}=\alpha,\ee
   for some constants $\gamma_1>4/3$ and $0<\alpha<+\infty$, then (\ref{negative}) holds for the following cases:\\
   (a) when $\gamma=4/3$ (where $\gamma$ is the constant in (\ref{whitedwarf1})),  if the  total mass $M$ is less than
    a ''critical mass'' $M_c$, then (\ref{negative}) holds. This case includes white dwarf stars. For a white dwarf star, $\gamma_1=5/3$. \\
    (b) When $\gamma>4/3$, (\ref{negative}) holds for arbitrary positive total mass $M$. This generalizes our previous result in \cite{LS1} for
    the polytropic stars with $p(\rho)=\rho^{\beta}$, $\beta>4/3$. \\
  It should be noted that (\ref{origin}) does not apply to suppermassive star, i.e. $p(\r)=k\r^{4/3}$. For the supermassive star, in order that
  have that (\ref{negative}) hold, in additional to requiring that the total mass is less than a  ''critical mass'',  we also require that the angular momentum
  (per unit mass) $J$ is not identically zero.

   \noindent ii) Assumption (2) in the above theorem implies that the functional $F$ is bounded below, i.e.,
\be\label{bounded below} \inf_{\r\in W_M} F(\r)>-\infty.\ee
We will verify this assumption in Section 4 (see Theorem 4.1).

 \noindent iii)  The inequality (\ref{2.15}) is crucial for the compactness result (\ref{2.17}). One of the difficulties in the analysis is the
 loss of compactness because we consider the problem in an unbounded space, $\RR^3$. The inequality (\ref{2.15}) means the masses  of the elements in  the minimizing sequence $T\r^i(x)$ ''almost'' concentrate in a ball $B_R(0)$.

 \noindent iv) It is easy to verify  that the functional $F$ is invariant under any vertical shift, i.e., if $\r(\cdot)\in W_M$, then
$\bar \r(x)=:\r(x+a{\bf e_3})\in W_M$ and $F(\bar \r)=F(\r)$ for any $a\in \RR$. Therefore, if $\{\r^i\}$ is a minimizing sequence of $F$ in
$W_M$, then $\{T\r^i\}=:=\r^i(x+a_i{\bf e_3})$  is also a minimizing sequence in $W_M$.
\end{rem}

Theorem \ref{aa} is proved in a sequence of lemmas with some modifications of the arguments in \cite{LS1}. We only  sketch the  proofs of those
  lemmas and Theorem \ref{aa}. Complete details can be followed as in \cite{LS1}. We first give some  inequalities which will be used later.  We begin with   Young's inequality (see \cite{GT}, p. 146.)
 \begin{lem} If $f\in L^p\cap L^r$, $1\le p<q<r\le +\infty$, then
\be\label{young} ||f||_q\le ||f||_p^a||f||_r^{1-a}, \qquad a=\f{q^{-1}-r^{-1}}{p^{-1}-r^{-1}}.\ee \end{lem}
The following two lemmas are proved in \cite{AB}.
\begin{lem}\label{bf1'} Suppose the function $f\in L^1(\RR^3)\cap L^{q}(\RR^3)$. If $1<q\le 3/2$, then $Bf=:f\ast\f{1}{|x|}$ is in $L^{r}(\RR^3)$ for
$3<r<3q/(3-2q)$, and
\be\label{bf1} || Bf||_r\le C \left(||f||_1^b||f||_q^{1-b}+||f||_1^c||f||_q^{1-c}\right),\ee for some constants $C>0$, $0<b<1$,  and $0<c<1$.
If $q>3/2$, then $Bf(x)$ is a bounded continuous function, and  satisfies (\ref{bf1}) with $r=\infty.$
\end{lem}
\begin{lem}\label{lem2.2} For any function $f\in L^1(\RR^3)\cap L^{4/3}(\RR^3)$, if $\gamma\ge 4/3$,  then $\nabla Bf\in L^2(\RR^3)$. Moreover,
\be\label{bf2} |\int f(x)Bf(x)dx|=\f{1}{4\pi}||\nabla Bf||_2^2\le C \left(\int|f|^{4/3}(x)dx\right)\left(\int|f|(x)dx\right)^{2/3},\ee for some constant $C$.
\end{lem}
We also need the following lemma.
\begin{lem}\label{lem3.1'} Suppose that the pressure function $p$ satisfies (\ref{whitedwarf1}) and (\ref{us}) holds. Let  $\{\r^i\}\subset W_M $ be a minimizing sequence for the functional $F$. Then there exists a constant $C>0$ such that
\be\label{ggg} \int [(\r^i)^{\gamma}(x)+\frac{1}{2}\f{\rho^i(x)L(m_{\r^i}(r(x)))}{r(x)^2}] dx\le C, {~\rm for~all~} i\ge 1,\ee
where $\gamma\ge 4/3$ is the constant in (\ref{whitedwarf1}).
So,  the sequence $\{\r^i\}$ is bounded in $L^{\gamma}(\RR^3)$.\end{lem}
\begin{proof} By (\ref{us}), we know that
\be\label{xx1} \int [A(\r^i)(x)+\frac{1}{2}\f{\rho^i(x)L(m_{\r^i}(r(x)))}{r(x)^2}] dx\le C, {~\rm for~all~} i\ge 1, \ee
for any minimizing sequence $\{\r^i\}\subset W_M $ for the functional $F$, where we have used that $\{F(\r^i)\}$ is bounded from above since it converges to $\inf_{W_M} F$.
It is easy to verify that, by virtue of (\ref{whitedwarf1}) and (\ref{z1}),
\be\label{whitedwarf1'}
  \lim_{\r\to \infty}\frac{A(\r)}{\r^{\g}}=\frac{K}{\gamma-1},\  A(\r)>0 {~\rm for~} \r>0.
 \ee
 Therefore, there exits   a constant $\r^*>0$  such that
 \be \alpha A(\r)\ge \r^{\gamma}, \qquad {\rm for~} \r\ge \r^*,
 \ee
where $\alpha=\frac{2(\gamma-1)}{K}$.
 Hence, for $\r\in W_M$,
 \begin{align}
 \int \r^{\gamma}dx&\le\int_{\r<\r^*}(\r^*)^{\gamma-1}\r dx+\alpha\int_{\r\ge \r^*}A(\r)dx\notag\\
 &\le (\r^*)^{\gamma-1}M+\alpha\int A(\r)dx.
 \end{align}
 Applying this inequality to $\r^i$, we conclude that the sequence $\{\r^i\}$ is bounded in $L^{\gamma}(\RR^3)$ by using (\ref{xx1}).
 \end{proof}
 For any $M>0$, we let
 \be\label{xx2} f_M=\inf_{\r\in W_M}F(\r).\ee
\begin{lem}\label{lem4.4} If (\ref{L1}) holds, then $f_{\bar M}\ge (\bar M/M)^{5/3}f_M$  for every $M>\bar M>0$ .\end{lem}
\begin{proof}

The proof follows from a scaling argument as in \cite{rein1} and \cite{LS1}. Take $a=(M/\bar M)^{1/3}$ and let $\bar \r(x)=\r(bx)$ for any
$\r\in W_M$. It is easy to verify that $\bar \r \in W_{\bar M}$.
Moreover, for $r\ge 0$, it is easy to verify, (as in \cite{LS1}) that
\be
m_{\b \r}(r)=\f{1}{a^3}m_{\r}(ar).\ee
Since $L$ satisfies (\ref{L1}) and $a> 1$,  we have
\be\label{L2}
L(m_{\bar \r}(r))\ge \f{1}{a^4} L(m_\r(br)).\ee
Thus, as in \cite{LS1}, we can show that
\be\label{012'}
\int\f{\b\r(x) L(m_{\b \r}(r(x)))}{r(x)^2}dx
=\f{1}{a^5}\int\f{\r(x) L(m_{\b \r}(r(x)))}{r(x)^2}dx.\ee
Therefore, since $a\ge 1$,  it follows from (\ref{L2}) and (\ref{012'}) that
\begin{align} F(\b\r)&\ge a^{-3}\int A(\r)dx-\f{a^{-5}}{2}\int \r B\r dx+\f{a^{-5}}{2}\int\f{\r(x) L(m_{\b \r}(r(x)))}{r(x)^2}dx\notag\\
&\ge a^{-5}\left(\int A(\r)dx-\f{1}{2}\int \r B\r dx+\f{1}{2}\int\f{\r(x) L(m_{\b \r}(r(x)))}{r(x)^2}dx\right)\notag\\
&=(\b M/M)^{5/3} F(\r).\end{align}
Since $\r\to\b\r$ is one-to-one between $W_M$ and $W_{\b M}$, this proves the lemma.

\end{proof}

\begin{lem}\label{lem4.5}  Let $\{\r^i\}\subset W_M $ be a minimizing sequence for $F$. Then there exist constants
$r_0>0$, $\delta_0>0$,  $i_0\in \mathbf{N}$ and $x^i\in \RR^3$ with $r(x^i)\le r_0$,  such that
\be\la{keynote}
\int_{B_1(x^i)}\r^i(x)dx\ge \delta_0, \ i\ge i_0.
\ee\end{lem}
\begin{proof}
First, since $\lim_{i\to\infty}F(\r^i)\to f_M$ and $f_M<0$ (see (3.4)), for large $i$,
\be\label{y1}
-\f{f_M}{2}\le -F(\r^i)\le \f{1}{2}\int \r^iB\r^idx.\ee
For any  $i$, let
\be \delta_i=\sup_{x\in \RR^3}\int_{|y-x|<1}\rho^i(y)dy.\ee
Now \begin{align}\label{y2'}
&\int \r^iB\r^i(x)dx\\
&=\int_{\RR^3}\r^i(x)\{\int_{|y-x|<1}+ \int_{1<|y-x|<r}+\int_{|y-x|>r}\}\f{\r^i(y)}{|y-x|}dydx\notag\\
&=:D_1+D_2+D_3,\end{align}
and $D_3\le M^2r^{-1}$. The shell $1<|y-x|<r$ can be covered by at most $ Cr^3$ balls of radius 1, so $D_2\le C M \delta_ir^3$.
By using H${\rm \ddot{o}}$lder's inequality and applying (\ref{bf1}) to the restriction of $\r^i$ to $\{y: |y-x|<1\}$, we get
\begin{align}\label{new111}
D_1&\le \|\r^i\|_{4/3}\|\int_{|y-x|<1}\f{\r^i(y)}{|y-x|}dy\|_4\notag\\
&\le C \|\r^i\|_{4/3}\left(\|\chi_{B_1(x)}\r^i\|_1^b\|\r^i\|_{4/3}^{1-b}+\|\chi_{B_1(x)}\r^i\|_1^c\|\r^i\|_{4/3}^{1-c}\right)\notag\\
&\le C \|\r^i\|_{4/3}\left(\delta_i^b\|\r^i\|_{4/3}^{1-b}+\delta_i^c\|\r^i\|_{4/3}^{1-c}\right),\end{align}
where  $0<b<1$ and $0<c<1$.  Now since  $\{\|\r^i\|_{\g}\}$ is bounded, it follows that $\{\|\r^i\|_{4/3}\}$ is bounded
due to the fact  $\g\ge 4/3$ in view of (\ref{young}) and  $\|\r^i\|_1=M$;
 this gives $D_1\le C(\delta_i^b+\delta_i^c)$. It follows that we could choose $r$ so large that the above estimates give $\int \r^iB\r^i(x)dx<-f_M$ {\it if $\delta_i$ were small enough}. This would contradict
(\ref{y1}). So there exists $\delta_0>0$ such that $\delta_i\ge \delta_0$ for large $i$. Thus, as $i$ is large, there
exists $x^i\in \RR^3$ and $i_0\in \mathbb{N}$ such that  \be\la{keynote1}
\int_{B_1(x^i)}\r^i(x)dx\ge \delta_0, \ i\ge i_0.
\ee   We now prove that there exists $r_0>0$ independent of $i$ such that  $x^i$ must satisfy
$r(x^i)\le r_0$ for $i$ large. Namely, since $\r^i$ has mass at least $\delta_0$ in the unit ball centered at $x^i$, and is
axially symmetric, it has mass $\ge Cr(x^i)\delta_0$ in the torus obtained by revolving this ball around $x_3$-axis (or $z$-axis).Therefore $r(x^i)\le (C\delta_0)^{-1}M.$
\end{proof}
In order to prove Theorem \ref{aa}, we will need the following lemma.
\begin{lem}\label{lem4.6} Let $\{f^i\}$ be a bounded sequence in $L^{\gamma}(\RR^3)$ ($\gamma\ge 4/3$) and suppose
$$f^i\rightharpoonup f^0~~ weakly~in~ L^{\gamma}(\RR^3).$$ Then\\
(a) For any $R>0$,
$$\nabla B(\chi_{B_R(0)}f^i)\to \nabla B(\chi_{B_R(0)}f^0) ~~strongly~in ~ L^2(\RR^3),$$ where $\chi$ is the indicator function.\\
(b) If in addition $\{f^i\}$ is bounded in $L^1(\RR^3)$, $f^0\in L^1(\RR^3)$, and for any $\epsilon>0$ there exist $R>0$ and $i_0\in \mathbf N$ such that
\be\la{y3}\int_{|x|>R}|f^i(x)|dx<\epsilon,\qquad i\ge i_0,\ee
then
 $$\nabla Bf^i\to \nabla Bf^0 ~strongly~in ~ L^2(\RR^3).$$
\end{lem}
\begin{proof}
This lemma follows easily from the proof of Lemma 3.7 in  \cite{rein1}, due to the following observation:\\
The map: $\r\in L^{\gamma}(\RR^3)\mapsto I_{B_R(0)}\nabla B(I_{B_R(0)}\r)$ is compact for any $R>0$, if $\gamma\ge 4/3$.
\end{proof}
 With above lemmas, the proof of Theorem \ref{aa} is similar to that in \cite{LS1}. So we only outline the main steps.
\vskip  0.2cm
\noindent {\it Proof of Theorem \ref{aa}}
\vskip 0.2cm
\noindent\underline{Step 1}.
We begin with a splitting as in \cite{rein1}.
For $\r\in W_M$, for any $0<R_1<R_2$, we have
\be\label{017}
\r=\r\chi_{|x|\le R_1}+\r\chi_{R_1<|x|\le R_2}+\r\chi_{|x|>R_2} =:\r_1+\r_2+\r_3,\ee where $\chi$ is the indicator function.
It is easy to verify that
\begin{align}\la{020} \int\f{\r(x)L(m_{\r}(r(x))}{r^2(x)}dx&=\sum_{j=1}^3\int\f{\r_j(x)L(m_{\r_j}(r(x))}{r^2(x)}dx\notag\\
&+\sum_{j=1}^3\int\f{\r_j(x)(L(m_{\r}(r(x))-L(m_{\r_j}(r(x))}{r^2(x)}dx,\notag\\
&\ge \sum_{j=1}^3\int\f{\r_j(x)L(m_{\r_j}(r(x))}{r^2(x)}dx. \end{align}
In the last inequality above, we have used (3.2). So, we have
\be\label{022}
F(\r)\ge \sum_{j=1}^3F(\r_j)-\sum_{1\le i<j\le 3}I_{ij},\ee
where
$$I_{ij}=\int_{\RR^3}\int_{\RR^3} |x-y|^{-1}\r_i(x)\r_j(y)dxdy,\qquad 1\le i<j\le 3.$$
If we choose $R_2>2R_1$ in the splitting (\ref{017}), then
\be\label{xxx1}I_{13}\le \frac{C}{R_2}. \ee
By (\ref{bf1}) and (\ref{bf2}), we have
\begin{align}\label{xxx2}
&I_{12}+I_{23}\notag\\
&=\frac{1}{4\pi}\int\nabla (B\rho_1+ B\r_3)\cdot \nabla B\r_2dx\le C \|\nabla (B\rho_1+ B\r_3)\|_2\|\nabla B\r_2\|_2\notag\\
&\le CM^{1/3} \|\r_1+\r_3\|_{4/3}^{2/3}\|\nabla B\r_2\|2\le C M^{1/3}\|\r\|_{4/3}^{2/3}\|\nabla B\r_2\|_2. \end{align}
Using  Lemma \ref{lem4.4}, (\ref{negative}), (\ref{022}), (\ref{xxx1}) and (\ref{xxx2}), and  following an argument as in the proof of Theorem 3.1 in \cite{rein1}, we can show that
\begin{align}
\label{rx1} &f_M-F(\rho)\notag\\
&\le (1-(\f{M_1}{M})^{5/3}-(\f{M_2}{M})^{5/3}-(\f{M_3}{M})^{5/3})f_M+C(R_2^{-1}+M^{1/3}\|\r\|_{4/3}^{2/3}||\nabla B\r_2||_2)\notag\\
&\le C f_M M_1M_3+C(R_2^{-1}+M^{1/3}\|\r\|_{4/3}^{2/3}||\nabla B\r_2||_2),\end{align}
by choosing $R_2>2R_1$ in the splitting (\ref{017}), where $M_i=\int \r_i(x)dx$ ($i=1, 2, 3$.)
 Let $\{\r^i\}$ be a minimizing sequence of $F$ in $W_M$. By Lemma \ref{lem4.5}, we know that there exists  $i_0 \in \mathbf{N}$ and $\delta_0>0$ independent of $i$ such that
 \be\label{y4}
\int_{a_i{\bf e_3}+B_{R_0(0)}}\r^i(x)dx\ge \delta_0,  \qquad  if~i\ge i_0,\ee
where $a_i=z(x^i)$ and $R_0=r_0+1$, $x^i$ and $r_0$ are those quantities in  Lemma \ref{lem4.5}, ${\bf e_3}=(0, 0,1)$. Having proved (\ref{y4}), we can  follow the argument in the proof of Theorem 3.1 in \cite{rein1} to verify (\ref{y3}) for
$$f^i(x)=T\r^i(x)=:\r^i(\cdot+a_i{\bf e_3})$$ by using (\ref{022}) and (\ref{y4}) and choosing
suitable $R_1$ and $R_2$ in the splitting (\ref{017}). We sketch this as follows. The sequence $T\r^i=:\r^i(\cdot+a_i{\bf e_3})$, $i\ge i_0$, is a minimizing sequence of $F$ in $W_M$ (see Remark 2 after Theorem \ref{aa}). We rewrite (\ref{y4}) as
\be\label{y4'}
\int_{B_{R_0}(0)}T\r^i(x)dx\ge \delta_0, \  i\ge i_0.\ee
 Applying (\ref{rx1}) with $T\r^i$ replacing $\r$,  and noticing that $\{T\r^i\}$ is bounded in $L^\g(\RR^3)$ (see Lemma 3.4) (so $\{\|T\r^i\|_{4/3}$\} is bounded if $\g\ge 4/3$ in view of (\ref{young}) and the fact $\|\r^i\|_1=M$),  we obtain, if $R_2>2R_1$,
\be\la{rx3} -C f_M M^i_{1}M^i_{3}\le C(R_2^{-1}+||\nabla BT\r^i_{2}||_2)+F(T\r^i)-f_M,\ee
where $M^i_{1}=\int T\r^i_1(x)dx=\int_{|x|<R_1}T\r^i(x)dx,$, $M^i_{3}=\int T\r^i_3(x)dx=\int_{|x|>R_2}T\r^i(x)dx$ and $T\r^i_{2}=\chi_{R_1<|x|\le R_2}T\r^i.$ Since $\{T\r^i\}$ is bounded in $L^{\g}(\RR^3)$, there exists a subsequence, still labeled by $\{T\r^i\}$, and a function $\tilde \r\in W_M$
such that $$T\r^i\rightharpoonup \tilde \r{~\rm weakly~ in~} L^{\g}(\RR^3).$$  This proves (\ref{2.16}). By (\ref{y4'}), we know that $M^i_{1}$ in (\ref{rx3}) satisfies $M^i_{1}\ge \delta_0$
for $i\ge i_0$ by choosing $R_1\ge R_0$ where $R_0$ is the constant in (\ref{y4'}). Therefore, by (\ref{rx3}) and the fact that $f_M<0$ (cf. (\ref{negative})) , we have
\be\la{rx4} -C f_M \delta_0 M^i_{3}\le CR_2^{-1}+C||\nabla B\tilde \r_{2}||_2+C||\nabla BT\r^i_{2}-\nabla B\tilde \r_{2}||_2)+F(T\r^i)-f_M,\ee
where  $\tilde \r_{2}=\chi_{|x|>R_2}\tilde \r$.  Given any $\epsilon>0$, by the same
argument as \cite{rein1}, we can increase $R_1>R_0$ such that the second term on the right hand side of (\ref{rx4}) is small, say less than $\epsilon/4$.
Next choose $R_2>2R_1$ such that the first term is small. Now that $R_1$ and $R_2$ are fixed, the third term on the right hand side of (\ref{rx4}) converges to zero by Lemma \ref{lem4.6}(a).  Since $\{T\r^i\}$ is a minimizing sequence of $F$ in $W_M$, we can make $F(T\r^i)-f_M$ small by taking $i$ large.
Therefore, for $i$ sufficiently large, we can make
\be\label{rx5}  M^i_{3}=:\int_{|x|>R_2}T\r^i(x)dx<\epsilon.\ee
This verifies (\ref{y3}) in Lemma 3.7 for $f^i=T\r^i$. By weak convergence we have that for any $\epsilon>0$ there exists $R>0$ such that
$$M-\epsilon\le \int_{B_R(0)}\tilde \r(x)dx\le M,$$
which implies $\tilde \r\in L^1(\RR^3)$ with $\int \tilde \r dx=M$. Therefore,  by Lemma \ref{lem4.6}(b),we have
\be\label{rx6} ||\nabla BT\r^i-\nabla B\tilde \r||_2\to 0, \qquad i\to +\infty.\ee
This proves (\ref{2.17}).    (\ref{2.15}) in Theorem \ref{aa} follows from (\ref{rx5}) by taking $R=R_2$.

 Let  $\{\r^i\}$ be a minimizing sequence of the energy functional  $F$, and let $\tilde \r$ be a weak limit  of $\{T\r^i\}$ in $L^{\gamma}(\RR^3)$.  We will prove that $\tilde \r$ is a minimizer of $F$ in $W_M$; that is
\be\label{2.80}
F(\tilde \r)\le \lim\inf_{i\to \infty} F(T\r^i).\ee
By (\ref{whitedwarf1}), there exist positive constants $C$ and $\r^*$ such that
\be\la{xxxx1} A'(\rho)\le C\rho^{\gamma-1}, for\  \rho\ge \rho^*, \ee
where $\gamma\ge 4/3$ is the constant in (3.3). Since $\tilde \rho\in L^{\gamma}$ and $\int\tilde \rho dx=M$,
we can conclude $A'(\tilde \rho)\in L^{\gamma'}$, where
$L^{\gamma'}$ is the dual space of $L^{\gamma}$, i.e.,  $\gamma'=\frac{\gamma}{\gamma-1}$.
In view of (2.5) and (3.3), we have
\be\label{xx3}
A''(\r)=p'(\r)/\r>0, \qquad {\rm for~} \r>0,  \ee
so that
\be\label{xxxx3}\int A(T\rho^i)dx\ge\int A(\tilde \rho)dx+\int A'(\tilde \rho)(T\rho^i-\tilde \rho), {\rm~ for~} i\ge 1. \ee
Since $A'(\tilde \rho)\in L^{\gamma'}$ and $T\rho^i$ weakly converges to $\tilde \rho$  in $L^{\gamma}$,
\be\label{xxxx4}\int A'(\tilde \rho)(T\rho^i-\tilde \rho)\to 0, {\rm~as~} i\to +\infty.\ee
Therefore,
\be\label{rx7} \int A(\tilde \r)dx\le \lim\inf_{i\to \infty} \int A(T\r^i)dx.\ee
 Next, following the proof in \cite{LS1}, we can show that
\be\label{rx20}
\lim_{i\to \infty}\inf\int\f{T\r^i(x)L(m_{T\r^i}(r(x))-\tilde \r(x)L(m_{\tilde \r}(r(x))}{r^2(x)}dx\ge 0,
\ee
by showing that the mass function
$$ m_{\tilde \r}(r)=:\int_{\sqrt {x_1^2+x_2^2}\le r} \tilde \r(x) dx$$
 is continuous for $r\ge 0,$
 and using (\ref{2.15}).  Then (\ref{2.80}) follows from (\ref{rx6}), (\ref{rx7}) and (\ref{rx20}).

\subsection{Stability}
In this section, we assume that the pressure function $p$ satisfies
\be\label{whitedwarf2}
 p\in C^1[0, +\infty),\ \lim_{\r\to 0+}\frac{p(\r)}{\r^{6/5}}=0,\  \lim_{\r\to \infty}\frac{p(\r)}{\r^{\gamma}}=K,\   p'(\r)>0  {~\rm for~} \r>0.
 \ee
 where  $0<K<+\infty$ and $\gamma\ge 4/3$ are constants.  It should be noticed that (\ref{whitedwarf2}) implies both (2.15) and (\ref{whitedwarf1}).
We consider the Cauchy problem for
(\ref{1.1}) with the initial data
\be\label{initial}\rho(x,0)=\rho_0(x),\
{\bf v}(x,0)={\bf v}_0(x).\ee
We begin by giving the definition of a weak solution.\\

\noindent {\bf Definition:}  Let $\r {\bf v}={\bf m}$. The triple $(\rho, {\bf m}, \Phi)(x, t)$ ($x\in\RR^3, t\in[0, T])$ $(T>0)$  and  $\Phi$ given by (\ref{phi}),  with $\r\ge 0,$ ${\bf m}$, ${\bf m}\otimes{\bf m}/\r$ and $ \r\nabla\Phi$ being in  $L^1_{loc}( \RR^3\times [0, T])$, is called a  {\it weak solution} of the Cauchy problem  (\ref{1.1}) and (\ref{initial}) on $ \RR^3\times [0, T]$  if for any Lipschitz continuous test  functions $\psi$ and ${\bf \Psi}=(\psi_1, \psi_2, \psi_3)$ with compact supports in $\RR^3\times [0, T]$,\\
  \be \int_0^T\int \left(\rho\psi_t+{\bf m}\cdot \nabla\psi\right)dxdt+\int\rho_0(x)\psi(x,0)dx=0,
 \ee
 and
 \be\label{3.3} \int_0^T\int \left({\bf m}\cdot{\bf \Psi}_t+\f{{\bf m}\otimes{\bf m}}{\rho}\cdot \nabla{\bf \Psi}\right)dxdt+\int{\bf m}_0(x){\bf \Psi}(x,0)dx=\int_0^T
 \int\rho\nabla \Phi{\bf \Psi} dxdt,
 \ee both hold.

\vskip 0.2cm

For any weak solution, it is easy to verify that the total mass is conserved by using a generalized divergence theorem for
 $L^{r}$ functions ($r\ge 1$) (cf. \cite{chenfrid2}),
 \be\label{5.1}
 \int\r(x, t)dx=\int \r(x, 0)dx,\qquad t\ge 0.\ee
The {\it total energy} of system (\ref{1.1}) at time $t$ is
\begin{equation}\label{energy}
E(t)=E(\r(t), {\bf v}(t))=\int\left(A(\r)+\frac{1}{2}\r|{\bf v}|^2\right)(x, t)dx-\frac{1}{8\pi}\int|\nabla \Phi|^2(x, t)dx,\end{equation}
where as before,
\be\label{A}A(\r)=\r\int_0^\r \f{p(s)}{s^2}ds.\ee
For a solution of (\ref{1.1}) without shock waves,  the total energy is conserved,
i.e., $E(t)=E(0)$ ($t\ge0$)(cf. \cite{Ta}). For  solutions with shock waves, the energy should be non-increasing in time,
so that for all $t\ge 0$,
\be\label{denergy} E(t)\le E(0),\ee
due to the entropy conditions, which are motivated by the second law of thermodynamics (cf. \cite{lax} and \cite{smoller}). This was  proved in
\cite{LS1}.

 We consider axi-symmetric initial data, which takes the form
 \begin{align}\la{5.2'}
  &\r_0(x)=\r(r, z),\notag\\
  & {\bf v}_0(x)=v^r_0(r, z){\bf e}_r+v^{\theta}_0(r, z){\bf e}_{\theta}+v^3_0(\r, z){\bf e}_3.
 \end{align}
  Here   $r=\sqrt {x_1^2+x_2^2},\ z=x_3$,  $x=(x_1, x_2,  x_3)\in \RR^3$ (as before), and
  \be {\bf e}_r=(x_1/r, x_2/r,  0)^\mathrm{T},\ {\bf e}_{\theta}=(-x_2/r,  x_1/r,\ 0)^\mathrm{T},\ {\bf e}_3=(0, 0, 1)^\mathrm{T}.\ee
  We seek axi-symmetric  solutions of the form
  \begin{align}\la{5.3'}
  &\r(x, t)=\r(r, z, t),\notag\\
  & {\bf v}(x, t)=v^r(r, z, t){\bf e}_r+v^{\theta}(r, z, t){\bf e}_{\theta}+v^3(r, z, t){\bf e}_3,\\
  &\Phi(x, t)=\Phi(r, z, t)=-B\r(r, z, t),
  \end{align}
 We call a vector field ${\bf u}(x, t)=(u_1, u_2, u_3)(x)$ ($x\in \RR^3$ ) axi-symmetric if it can be written
in the form
$${\bf u}(x)=u^r(r, z){\bf e}_r+u^{\theta}(r, z){\bf e}_{\theta}+u^3(\r, z){\bf e}_3.$$
For the velocity field ${\bf v}=(v_1, v_2, v_3)(x, t)$, we define the angular momentum (per unit mass)  $j(x,t)$ about the $x_3$-axis  at $(x, t)$ ,  $t\ge 0$, by
\be\la{5.3}
j(x, t)=x_1v_2-x_2v_1.\ee
For an axi-symmetric velocity field
\be\la{asv}
{\bf v}(x, t)=v^r(r, z, t){\bf e}_r+v^{\theta}(r, z, t){\bf e}_{\theta}+v^3(\r, z, t){\bf e}_3,\ee
\be\la{comp}
v_1=\f{x_1}{r}v^r-\f{x_2}{r}v^{\theta},\ v_2=\f{x_2}{r}v^r+\f{x_1}{r}v^{\theta}, v_3=v^3,\ee
so that
\be\la{j}j(x, t)= r v^{\theta}(r, z,  t). \ee
In view of ( {\ref{asv}) and (\ref{j}), we have
\be\label{V}
|{\bf v}|^2=|v^r|^2+\f{j^2}{r^2}+|v^3|^2.\ee
Therefore, the total energy at time $t$ can be written as
\begin{align}\label{en}
E(\r(t), {\bf v}(t))
&=\int A(\r)(x, t)dx+\frac{1}{2}\int \f{\r j^2(x,t)}{r^2(x)}dx\notag\\
&-\frac{1}{8\pi}\int|\nabla B\r|^2(x, t)dx+\frac{1}{2}\int \r(|v^r|^2+|v^3|^2)(x, t)dx.\end{align}

There are two important conserved quantities for the Euler-Poisson equations (\ref{1.1}); namely  the total mass and the angular momentum. In order to
describe these, we define $D_t$,  the non-vacuum region at time $t\ge 0$ of the solution by
\be\label{nonvacuum}
D_t=\{x\in \RR^3: \r(x, t)>0\}.
\ee
We will make the following physically reasonable assumptions A1)-A4) on weak solutions of the Cauchy problem (\ref{1.1}) and (\ref{initial})
(A1)-A4) are easily verified for smooth solutions. For general weak solutions, they are motivated by physical considerations, cf.\cite{Ta}).
\vskip 0.2cm
 A1) For any $t\ge0$, there exists a measurable subset $G_t\subset D_t$ with $meas(D_t-G_t)=0$ ($meas$ denotes  Lebsegue measure)  such that,
for any $x\in G_t$, there exists a unique (backwards) particle path $\xi(\tau, x, t)$ for $0\le \tau\le t$ satisfying

\be\label{particlepath}
\pa_{\t}\xi(\t, x, t)={\bf v}(\xi(\t, x, t), \t),\ \xi(t, x, t)=x.\ee

\begin{rem} If ${\bf v}(\cdot,t)\in BV(\RR^3)$ and $div_x {v}(\cdot, t)\in L^{\infty}(\RR^3)$ for $t\ge 0$ ($div_x$ is in the sense of distributions),
it was proved by L. Ambrosio (\cite{Ambrosio}) that A1) is valid. Related results can be found in \cite{lp}. \end{rem}

\vskip 0.2cm

For  $x\in G_t$, we write $$\xi(0, x, t)=\xi_{-t}(x).$$ Also, for $x\in \RR^3$ and $t\ge 0$,  we denote the total mass at time $t$
in the cylinder $\{y\in \RR^3: r(y)\le r(x)\}$ by $m_{\r(t)}(r(x))$, i.e.,
\be\label{mass}
m_{\r(t)}(r(x))=\int_{r(y)\le r(x)}\r(y, t)dy.\ee
For  axi-symmetric motion, we assume

\vskip 0.2cm

A2)
\be\label{mass1}
m_{\r(t)}(r(x))=m_{\r_0}(r(\xi_{-t}(x))), \qquad {\rm for~}  x\in G_t, t\ge 0.\ee
(This means that the mass enclosed within any material volume cannot change as we follow the volume in its motion ( \cite{Ta}, p. 47)).
Moreover, we assume that the angular momentum  is conserved along the particle path:

\vskip 0.2cm

A3) \be\label{angular1}j (x, t)=j( \xi_{-t}(x), 0), \qquad {\rm for~}  x\in G_t, t\ge 0.\ee

\vskip 0.2cm

\noindent Finally, for $L=j^2$,  we need a technical assumption; namely, \\
A4) \be\label{extra1}
\lim_{r\to 0+}\frac{L(m_{\r(t)}(r)+m_{\tilde \r}(r))m_{\sigma(t)}(r)}{r^2}=0,
\ee
for $t\ge 0$, where $\sigma(t)=\r(t)-\tilde \r. $
\begin{rem} (\ref{extra1})  can be understood as follows. For any $\r\in W_M$, we have $\lim_{r\to 0+} m_{\r}(r)=0. $ Therefore $\lim_{r\to 0+}L(m_{\r(t)}(r)+m_{\tilde \r}(r))=L(0)=0,$
so if we define $$\hat \r(s, t)-\hat \tilde \r(s)=\int_{-\infty}^{+\infty} (\r(s,z, t)- \tilde \r(s,z))dz, $$
then if
\be\label{good5} \f{m_{\sigma(t)}(r)}{r^2}=\f{\int_0^r(2\pi s (\hat \r(s, t)-\hat {\tilde \r}(s))ds}{r^2}\in L^{\infty}(0, \delta) \  for\  some \  \delta>0,
\ee
  (\ref{extra1}) will hold.
If $\hat \r(\cdot, t)-\hat \tilde \r(\cdot)\in L^{\infty}(0, \delta)$, then  (\ref{good5}) holds. This can be assured by assuming that
$\r(r, z, t)-\tilde \r(r, z)\in L^{\infty}((0, \delta)\times \RR\times \RR^+)$ and decays  fast enough in the $z$ direction. For example,
when $\r(x, t)-\tilde \r(x)$ has  compact support in $\RR^3$ and $\r(\cdot, t)-\tilde \r(\cdot)\in L^{\infty}(\RR^3)$, then (\ref{extra1}) holds.
\end{rem}

\vskip 0.2cm

We next make some assumptions on the initial data; namely,  we assume that the initial data is such that the initial total mass and
angular momentum are the same as those of the rotating star solution (those two quantities are conserved quantities). Therefore,
we require
\vskip 0.2 cm
I$_1$) \be\label{initial mass}
\int \r_0(x)dx=\int \tilde \r(x)dx=M. \ee
Moreover we assume
\vskip 0.2cm

I$_2$) For the initial angular momentum $j (x, 0)=rv_0^{\theta}(r, z)=: j_0(r, z)$ ($r=\sqrt {x_1^2+x_2^2}$, $z=x_3$ for $x=(x_1, x_2, x_3)$,
we assume
$j(x, 0)$ only depends on the total mass in the cylinder $\{y\in\RR^3, r(y)\le r(x)\}$, i.e. ,
\be\label{ia}
j(x, 0)=j_0\left(m_{\r_0}(r(x))\right).\ee
(This implies that we require that $v_0^{\theta}(r, z)$ only depends on $r$.)\\
Finally, we assume that the initial profile of the angular momentum per unit mass is the same as that of the rotating star solution, i. e.,
\vskip 0.2cm
 I$_3$) \be\label{ia1}
j_0^2(m)=L(m), \qquad 0\le m\le M,\ee
where $L(m)$ is the profile of the square of the angular momentum of the rotating star defined in Section 2.\\
In order to state our stability result, we need some notation.
Let $\lambda$ be the constant in Theorem 2.2, i.e.,
\be\label{lam}
\begin{cases}
& A'(\tilde \r(x))+\int_{r(x)}^{\infty}L(m_{\tilde \r}(s))s^{-3}ds-B\tilde \r(x)=\lambda, \  x\in \Gamma,\\
&\int_{r(x)}^{\infty}L(m_{\tilde \r})(s))s^{-3}ds-B\tilde \r(x)\ge \lambda, \qquad x\in \RR^3-\Gamma,\end{cases}\ee
with $A$ defined in (\ref{A}) and  $\Gamma$ defined in  (2.16).\\

For $\r\in W_M$, we define, \be d(\r, \tilde \r)=\int
[A(\rho)-A(\tilde \r)]
+(\r-\tilde \r)\int_{r(x)}^{\infty}\{\f{L(m_{\tilde \r}(s))}{s^3}ds-\lambda-B\tilde \r\}dx.
\ee
For $x\in \Gamma$, in view of the convexity of the function $A$ (cf. (\ref{xx3})) and (\ref{lam}), we have,
\begin{align}
&(A(\rho)-A(\tilde \r))(x)
+(\int_{r(x)}^{\infty}\f{L(m_{\tilde \r}(s))}{s^3}ds-\lambda-B\tilde \r(x))(\r-\tilde \r)\notag\\
&= (A(\r)-A(\tilde \r)-A'(\tilde \r)(\r-\tilde \r))(x)\ge 0.
\end{align}
For $x\in \RR^3-\Gamma$, $\tilde \r(x)=0$, so we have $A(\tilde \r)(x))=0$. This is because  since $A(0)=0$ due to $p(0)=0$ (cf. (3.3)) and (2.5).
Therefore, by (\ref{lam}), we have, for $\r\in W_M$ and  $x\in \RR^3-\Gamma$,
\begin{align}
&(A(\rho)-A(\tilde \r))(x)
+(\int_{r(x)}^{\infty}\f{L(m_{\tilde \r}(s))}{s^3}ds-\lambda-B\tilde \r(x))(\r-\tilde \r)\notag\\
&= A(\r)\ge 0.
\end{align}
Thus, for $\r\in W_M$,
\be d(\r, \tilde \r)\ge 0. \ee
We also define
\begin{align}\label{d1}d_1(\r, \tilde \r)
&=\f{1}{2}\int\f{\r(x) L(m_{\r}(r(x))-\tilde \r(x) L(m_{\tilde \r}(r(x))}{r^2(x)}dx\notag\\
&-\int \int_{r(x)}^{\infty}s^{-3}L(m_{\tilde \r}(s))ds(\r(x)-\tilde \r(x))dx,
\end{align}
for $\r\in W_M$. We shall show later that $d_1\ge 0$.
Our main stability result in this paper is the following global-in-time stability theorem.
\begin{thm}\label{th5.1} Suppose that the pressure function satisfies (\ref{whitedwarf2}),  and both (\ref{negative}), (\ref{us}) hold.  Let $\tilde \r$ be a minimizer of the functional $F$ in $W_M$, and assume that it  is unique up to a vertical shift.  Assume that I$_1$)- I$_3$), [(\ref{initial mass})-(\ref{ia1})] hold. Moreover,  assume that the angular momentum  of the rotating star solution $\tilde \r$ satisfies  (\ref{L}), (\ref{L1}) and (3.2).  Let $(\r, {\bf v}, \Phi)(x, t)$ be an axi-symmetric weak solution of the Cauchy problem (\ref{1.1}) and (\ref{initial}) satisfying the  assumptions A1)-A4), [(\ref{particlepath})-(\ref{extra1})].   If the total energy $E(t)$ (cf. (\ref{energy})) is non-increasing with respect to $t$,
then for every $\epsilon>0$, there exists a number $\delta>0$ such that if
\begin{align} &d(\r_0, \tilde \r)+\f{1}{8\pi}||\nabla B\r_0-\nabla B\tilde \r||_2^2+ |d_1(\r_0, \tilde \r)|\notag\\
&+\f{1}{2}\int \r_0(x)(|v^r_0|^2+|v^3_0|^2)(x)dx
<\delta,\end{align}
then  there is a vertical shift $a{\bf e_3}$ ($a\in \RR$, ${\bf e_3}=(0, 0, 1)$) such that, for every $t>0$
\begin{align} &d(\r(t), T^a\tilde \r)+\f{1}{8\pi}||\nabla B\r(t)-\nabla BT^a\tilde \r||_2^2+|d_1(\r(t), T^a\tilde \r)|\notag\\
&+\f{1}{2}\int \r(x, t)(|v^r(x, t)|^2+|v^3(x, t)|^2)dx
<\epsilon,
\end{align}
where $T^a\tilde \r(x)=:\tilde \r(x+a{\bf e_3}).$
\end{thm}
\begin{rem} As noted in \cite{LS1}, the vertical shift $a{\bf e_3}$ appearing in the theorem is analogous to a similar phenomenon which appears in the study of stability of viscous traveling waves in conservation laws, whereby convergence is to a ``shift`` of the original traveling wave.
\end{rem}
\begin{rem}Without the uniqueness assumption for the minimizer of $F$  in $W_M$, we can have the following type of stability result, as observed
in \cite{Rein} for the non-rotating star solutions. Suppose the assumptions in Theorem \ref{th5.1} hold.  Let $\mathcal{S}_M$ be the set of all minimizers of $F$ in $W_M$ and  $(\r, {\bf v}, \Phi)(x, t)$ be an axi-symmetric weak solution of the Cauchy problem (\ref{1.1}) and (\ref{initial}). If the total energy $E(t)$ is non-increasing with respect to $t$,
then for every $\epsilon>0$, there exists a number $\delta>0$ such that if
\begin{align} &\inf_{\tilde \r\in \mathcal{S}_M}\left[ d(\r_0, \tilde \r)+\f{1}{8\pi}||\nabla B\r_0-\nabla B\tilde \r||_2^2+ |d_1(\r_0, \tilde \r)|\right]\notag\\
&+\f{1}{2}\int \r_0(x)(|v^r_0|^2+|v^3_0|^2)(x)dx
<\delta,\end{align}
then for every $t>0$
\begin{align} &\inf_{\tilde \r\in \mathcal{S}_M}\left[d(\r(t), T^a\tilde \r)+\f{1}{8\pi}||\nabla B\r(t)-\nabla BT^a\tilde \r||_2^2+|d_1(\r(t), T^a\tilde \r)|\right]\notag\\
&+\f{1}{2}\int \r(x, t)(|v^r(x, t)|^2+|v^3(x, t)|^2)(x)dx
<\epsilon.
\end{align}
then  there is a vertical shift $a{\bf e_3}$ ($a\in \RR$, ${\bf e_3}=(0, 0, 1)$) such that, for every $t>0$
\begin{align} &d(\r(t), T^a\tilde \r)+\f{1}{8\pi}||\nabla B\r(t)-\nabla BT^a\tilde \r||_2^2+|d_1(\r(t), T^a\tilde \r)|\notag\\
&+\f{1}{2}\int \r(x, t)(|v^r(x, t)|^2+|v^3(x, t)|^2)dx
<\epsilon,
\end{align}
where $T^a\tilde \r(x)=:\tilde \r(x+a{\bf e_3}).$
 \end{rem}

The proof of  Theorem \ref{th5.1} follows from several lemmas. The proofs of these lemmas are similar to those in \cite{LS1}, and  therefore
we only sketch them.  First we have
\begin{lem}\label{lem5.2}
Suppose the angular momentum of the rotating star solutions satisfies  (\ref{L}), (\ref{L1}) and (3.2).  For any $\r(x)\in W_M$, if
\be\label{extra}
\lim_{r\to 0+}{L(m_\r(r)+m_{\tilde \r}(r))m_{\sigma}(r)}{r^{-2}}=0,
\ee
where $\sigma=\r-\tilde \r,$
then
\be\label{dd1} d_1(\r, \tilde \r)\ge 0,
\ee
where $d_1$ is defined by (\ref{d1}).
\end{lem}
\begin{proof}  For an axi-symmetric  function $f(x)=f(r, z)$ ($r=\sqrt {x_1^2+x_2^2},\ z=x_3$ for $x=(x_1, x_2, x_3)$),  we let
\be \hat f(r)=2\pi r\int_{-\infty}^{+\infty} f(r, z)dz,\ee

\be\label{xd1} m_f(r)=\int_{\{x: \sqrt{x_1^2+x_2^2}\le r\}}f(x)dx=\int_0^r \hat f(s)ds,\ee
so that
\be\label{dx2}
m'_f(r)=\hat f(r).\ee
In order to show (\ref{dd1}), we let
\be \sigma(x)=(\r-\tilde \r)(x),\ee
and for $0\le \alpha\le 1$, we define
\begin{align}
Q(\alpha)&=\f{1}{2}\int\f{(\tilde \r+\alpha\sigma)(x) L(m_{\tilde \r+\alpha\sigma}(r(x)))-\tilde \r(x) L(m_{\tilde \r}(r(x)))}{r^2(x)}dx\notag\\
&-\alpha\int \int_{r(x)}^{\infty}s^{-3}L(m_{\tilde \r}(s))ds\sigma(x)dx.
\end{align}
 Then
\be\label{5.37} Q(0)=0,\ Q(1)=d_1(\r,\ \tilde \r). \ee
By the assumption that $L'(m)\ge 0$ for $0\le m\le M$ (cf. (3.2)) and (\ref{extra}), we can show that
\be\label{dx8}
Q'(\alpha)=\int_0^{+\infty}\hat \sigma(r)\int_{r}^{\infty}s^{-3}(L(m_{\tilde \r+\alpha\sigma}(s))-L(m_{\tilde \r}(s)))dsdr,\ee
and therefore
\be\label{dx9} Q(0)=Q'(0)=0.\ee
This is done by interchanging the order of integration and integrating  by parts (details can be found in \cite{LS1}).
Differentiating (\ref{dx9}) again and interchanging the order of integration, we get
\be\label{dx101}
\f{d^2Q(\alpha)}{d\alpha^2}=\alpha\int_0^{+\infty}s^{-3}L'(m_{\tilde \r+\alpha\sigma}(s))(m_\sigma(s))^2ds.\ee
Therefore, if $L'(m)\ge 0$ for $0\le m\le M$, then
\be\label{dx11}\f{d^2Q(\alpha)}{d\alpha^2}\ge 0,\  for \ 0\le\alpha\le 1.\ee
This, together with (\ref{dx9})and (\ref{5.37}),  yields $d_1(\r, \tilde \r)=Q(1)\ge 0.$
 \end{proof}

\begin{lem}\label{lem5.3} Let $(\r, {\bf v})$ be a solution of the Cauchy problem (\ref{1.1}), (\ref{initial}) as stated in Theorem 3.2, then
\begin{align}\label{ed}
&E(\r, {\bf v})(t)-F(\tilde \r)\notag\\
&=d(\r(t),\tilde \r)+d_1(\r(t), \tilde \r)
-\f{1}{8\pi}||\nabla B\r(\cdot, t)-\nabla B\tilde \r||_2^2\notag\\
&+\f{1}{2}\int\r (|v^r|^2+|v^3|^2)(x, t)dx.\end{align}\end{lem}
\begin{proof}
From A1)-A3), we can show
\be\label{jj3}
j^2(x, t)=L(m_{\r(t)}(r(x))),\qquad x\in G_t.\ee
Therefore, by (\ref{en}), we have
\begin{align}\label{jj4}
E(\r(t), {\bf v}(t))
&=\int A(\r)(x, t)dx+\frac{1}{2}\int \f{\r(x, t) L(m_{\r(t)}(r(x))}{r^2(x)}dx\notag\\
&-\frac{1}{8\pi}\int|\nabla B\r|^2(x, t)dx+\frac{1}{2}\int \r(|v^r|^2+|v^3|^2)(x, t)dx.\end{align}
(\ref{ed}) follows from (\ref{jj4}) and the follow identities:
\begin{align*}
&(||\nabla B\r(\cdot, t)||_2^2-||\nabla B\tilde \r||_2^2)\notag\\
&=||\nabla (B\r(\cdot, t))-\nabla B\tilde \r)||_2^2+2\int \nabla B\tilde \r(x)\cdot (\nabla B\r(x, t)-\nabla B\tilde \r(x))dx\notag\\
&=||\nabla (B\r(\cdot, t))-\nabla B\tilde \r)||_2^2-8\pi\int  B\tilde \r(x) (\r(x, t)- \tilde \r(x))dx. \end{align*}
and   $$\int \r(x, t)dx=\int\tilde \r(x)dx=M.$$ \end{proof}

\noindent Having established these lemmas, the proof of Theorem 3.2 is the same as the proof of Theorem 3.1 in \cite{LS1}.

\section{Applications to White Dwarf and Supermassive Stars}
In this section, we want to verify the assumptions (\ref{negative}) and (\ref{us}) in Theorem 3.1 for both white dwarfs and supermassive stars.
Once we verify (\ref{negative}) and (\ref{us}), we can apply Theorems 3.1 and 3.2.  We begin with the following theorem which verifies (3.5) for white dwarfs, supermassive stars, and polytropes with $\gamma\ge 4/3$, in both the rotating and non-rotating cases.
\begin{thm} Assume that the pressure function $p$ satisfies (3.3). Then there exists a constant $\mathfrak{M}_c$ satisfying $0<\mathfrak{M}_c<\infty$ if $\gamma=4/3$ and $\mathfrak{M}_c=\infty$ if $\gamma>4/3$,  such that if $M<\mathfrak{M}_c$, then (\ref{us}) holds for $\r\in W_M$.
\end{thm}
\begin{proof} Using (\ref{bf2}), we have, for $\r\in W_M$,
\begin{align}\la{00}
F(\r)&=\int
  [A(\rho)+\f{1}{2}\f{\rho(x)L(m_{\r}(r(x)))}{r(x)^2}-\frac{1}{2}\rho B\rho]dx\notag\\
  &\ge \int
  [A(\rho)+\f{1}{2}\f{\rho(x)L(m_{\r}(r(x)))}{r(x)^2}]dx -C\int \r^{4/3}dx\left(\int \r\ dx\right)^{2/3}\notag\\
  &=\int
  [A(\rho)+\f{1}{2}\f{\rho(x)L(m_{\r}(r(x)))}{r(x)^2}]dx-CM^{2/3}\int \r^{4/3}dx.
  \end{align}
   Taking $p=1$, $q=4/3$, $r=\g$,  and $a=\f{\f{3}{4}\g-1}{\g-1}$  (where $\g\ge 4/3$ is the constant in (3.3)) in Young's inequality (\ref{young}), we obtain,
  \be ||\r||_{4/3}\le ||\r||_1^a||\r||_{\g}^{1-a}=M^a||\r||_{\g}^{1-a}.\ee
  This, together with (3.16)-(3.18) yields
  \begin{align}\label{haha}\int\r^{4/3}dx&\le M^{\f{4}{3}a}(\int\r^\g dx)^b\le M^{\f{4}{3}a}\left((\r^*)^{\gamma-1}M+\alpha\int A(\r)dx\right)^b\notag\\
  &\le C\left(M^{\f{4}{3}a+b}(\r^*)^{1/3}+\alpha M^{\f{4}{3}a}(\int A(\r)dx)^b\right),
  \end{align}
  where $b=\f{1}{3(\g-1)}$, $\alpha$ and $\r^*$ are the constants in (3.17) and we have used the elementary inequality $(x+y)^b\le C(x^b+y^b)$, for
  $x,\ y >0,\  0<b< 1$, for some constant $C$.   Therefore, (\ref{00}) and (\ref{haha}) imply
  \be\la{0001}  \int
  [A(\rho)+\f{1}{2}\f{\rho(x)L(m_{\r}(r(x)))}{r(x)^2}]dx\le F(\r)+ C\alpha M^{\f{4}{3}a+\f{2}{3}}(\int A(\r)dx)^b+CM^{\f{4}{3}a+b+\f{2}{3}}(\r^*)^{1/3}.\ee
  If $\g>4/3$, then  $0<b<1$, if $\g=4/3$, then $b=1$.
   Therefore (\ref{0001}) implies (\ref{us}).\end{proof}
 The next result shows that (3.4) holds for a wide class of (rotating or non-rotating) stars, including White Dwarfs.
\begin{thm}\label{th44.2} Suppose that the pressure function $p$ satisfies (\ref{whitedwarf1}) and
   \be\label{origin'} \lim_{\r\to 0+}\frac{p(\rho)}{\r^{\gamma_1}}=\beta,\ee
   for some constants $\gamma_1>4/3$ and $0<\beta<+\infty$, and assume that the angular momentum (per unit mass) satisfies (2.9). Then
   there exists $\mathbb{M}_c$ satisfying $0<\mathbb{M}_c<+\infty$ if $\g=4/3$ and $\mathbb{M}_c=+\infty$ if $\g>4/3$ such that
   if $M<\mathbb{M}_c$, then (\ref{negative}) holds, where $\g$ is the constant in (\ref{whitedwarf1}). \end{thm}
   \begin{rem} White dwarfs satisfy (\ref{whitedwarf1}) and (\ref{origin'}) with $\g=4/3$ and $\g_1=5/3$. \end{rem}
   \vskip 0.2cm
    \noindent {\it Proof of Theorem \ref{th44.2}}\\
     Due to (\ref{whitedwarf1}) and (\ref{origin'}), we can apply Theorem 2.1. Let $\hat \r(x)\in W_{M, S}$ be a minimizer $F(\rho)$ in  $W_{M, S}$ as
     described in Theorem 2.1, and let
$$G=\{x\in \RR^3:\  \hat \r(x)>0\}. $$
 Then $\bar G$ is a compact set in $\RR^3$, and $\hat \r\in C^1(G)$. Furthermore, there exists a constant $\mu<0$
such that
\be\label{lambda1'}
\begin{cases}
& A'(\hat \r(x))+\int_{r(x)}^{\infty}L(m_{\hat \r}(s)s^{-3}ds-B\hat \r(x)=\mu, \qquad x\in G,\\
&\int_{r(x)}^{\infty}L(m_{\hat \r}(s)s^{-3}ds-B\hat \r(x)\ge \mu, \qquad x\in \RR^3-G.\end{cases}\ee
     It follows from \cite{AB} that there exists  $\hat \r\in W_{M, S}\subset W_M$ such that $F(\hat \r)=\inf_{\r\in W_{M, S}}F(\r)$. It is easy to verify that
the triple $(\hat \r, \hat {\bf v}, \hat\Phi)$ is a time-independent solution of the Euler-Poisson equations (\ref{1.1}) in the region $G=\{x\in \RR^3:\  \hat \r(x)>0\},$  where $\hat {\bf v}=(-\f{x_2 J(m_{\hat \r}(r))}{r},  \f{x_1 J(m_{\hat \r}(r))}{r}, 0)$ and $\hat \Phi=-B\hat \r$.
Therefore
\be\label{04}\nabla_x p(\hat \r)=\hat \r\nabla_x(B\hat \r)+\hat \r L(m_{\hat \r})r(x)^{-3}{\bf e}_r, \ x\in G,
\ee
where ${\bf e}_r=(\f{x_1}{r(x)}, \f{x_2}{r(x)}, 0)$. Moreover, it is proved in \cite{CF} that  the boundary $\pa G$ of $G$  is smooth enough
to apply the Gauss-Green formula  on G. Applying the Gauss-Green formula on G and noting that $\hat \r|_{\pa G}=0$, we obtain,
\be\label{05}
\int_G x\cdot \nabla_x p(\hat \r)dx=-3\int_G p(\hat \r)dx=- 3\int p(\hat \r)dx.\ee
As in \cite{LS1}, we have
\be\label{jjyy} \int_G x\cdot \hat \r\nabla_x B\hat \r dx=-\f{1}{2}\int_G\hat \r B\hat \r dx=-\f{1}{2}\int\hat \r B\hat \r dx.
\ee
Next, since $x\cdot {\bf e}_r=r(x)$, we have
\begin{align}\label{07}
&\int_G x\cdot \hat \r(x) L(m_{\hat \r}(r(x))r^{-3}(x){\bf e}_rdx\notag\\&=\int_G \hat \r(x) L(m_{\hat \r}(r(x))r^{-2}(x)dx\notag\\&=\int \hat \r(x) L(m_{\hat \r}(r(x))r^{-2}(x)dx.\end{align}
Therefore,  from (\ref{05})-(\ref{07}) we have
\be\label{ppxx1}
-3\int p(\hat \r)dx=-\f{1}{2}\int\hat \r B\hat \r dx+\int \hat \r(x) L(m_{\hat \r}(r(x))r^{-2}(x)dx.\ee
Let $\bar \r(x)=b^3\hat \r(bx),$ for $b>0$;  then  $\bar \r\in W_M$. Also,
it is easy to verify  that the following identities hold,
\begin{align}\label{010}
&\int \bar \r B\bar \r dx=\int_{\RR^3}\int_{\RR^3}\f{\bar \r(x)\bar \r(y)}{|x-y|}dxdy\notag\\
&=b\int\int_{\RR^3}\int_{\RR^3}\f{\hat \r(x)\hat \r(y)}{|x-y|}dxdy =b\int \hat \r B\hat \r dx
\end{align}
\be\label{0101}
\int  A(\bar\r )dx=b^{-3}\int  A(b^3 \hat\r(x)) dx.
\ee
Moreover, for $r\ge 0$,
\begin{align}
m_{\b \r}(r)&=2\pi \int_0^r s\int_{-\infty}^{\infty} \b \r(s, z)dsdz\notag\\
&=2\pi \int_0^r s\int_{-\infty}^{\infty}  \hat \r(bs, bz)dsdz\notag\\
&=2\pi \int_0^{br} s'\int_{-\infty}^{\infty}  \r(s', z')ds'dz'\notag\\
&=m_{ \r}(br).\end{align}
Therefore,
\begin{align}\label{012}
\int\f{\b\r(x) L(m_{\b \r}(r(x)))}{r(x)^2}dx&= \int\f{b^3\hat\r(x) L(m_{\hat \r}(b r(x)))}{r(x)^2}dx\notag\\
&=b^2\int\f{\hat\r(x) L(m_{\hat \r}(r(x)))}{r(x)^2}dx. \end{align}
 It follows from (\ref{010})-(\ref{012}) that
\begin{align}\label{ppxx2} F(\b\r)&= b^{-3}\int A(b^3\hat \r)dx-\f{1}{2}b \int \hat \r B\hat \r dx\notag\\&+\f{b^2}{2}\int\f{\hat \r(x) L(m_{ \hat \r}(r(x)))}{r(x)^2}dx.\end{align}
Hence, (\ref{ppxx1}) and (\ref{ppxx2}) give
\begin{align}\label{ppxx3}
F(\b\r)&= \int \left(b^{-3} A(b^3\hat \r)-3bp(\hat \r(x))\right)dx\notag\\&
+\left(\f{b^2}{2}-b\right)\int\f{\hat \r(x) L(m_{ \hat \r}(r(x)))}{r(x)^2}dx.\end{align}
In view of (2.9), we have
\be\label{ppxx4}\left(\f{b^2}{2}-b\right)\int\f{\hat \r(x) L(m_{ \hat \r}(r(x)))}{r(x)^2}dx\le 0,\ee
if $b>0$ is small. It follows from (\ref{origin}) that
\be\label{ppxx5} \f{1}{2}\beta \r^{\g_1}\le p(\r)\le 2\beta \r^{\g_1}, \ {\rm for~small~} \rho.\ee
Thus, when $b$ is small, since $\hat \r$ is bounded, we have
\be\label{ppxx5'} \f{\beta}{2(\g_1-1)}b^{3\g_1}(\hat\r)^{\g_1}(x)\le A(b^3\hat \r(x))\le  \f{2\beta}{\g_1-1}b^{3\g_1}(\hat\r)^{\g_1}(x),\ee
for $x\in \RR^3$. Hence, (\ref{ppxx4}) and (\ref{ppxx5}) imply
\begin{align}\label{ppxx6}
&\int \left(b^{-3} A(b^3\hat \r)-3bp(\hat \r(x))\right)dx\notag\\
&\le \beta \int \left( \f{2}{\g_1-1}b^{3\g_1-3}-\f{3}{2}\right)(\hat\r)^{\g_1}dx. \end{align}
Since $\gamma_1>4/3$, we have $3\g_1-3>1$. Therefore, we conclude that
\be\label{ppxx7}\int \left(b^{-3} A(b^3\hat \r)-3bp(\hat \r(x))\right)dx<0,\ee
for small $b$.  (\ref{negative}) follows from  (\ref{ppxx3}), (\ref{ppxx4}) and (\ref{ppxx7}). This completes the proof of Theorem 4.2.$\Box$
\vskip 0.2cm

We show next that if the angular momentum distribution is everywhere positive, we
may apply the existence theorem of Friedman and Tarkington, \cite{FT2}, to conclude
that (3.4) holds with no total mass restriction.  This result applies also to White Dwarfs.

\begin{thm}\label{thm44.3} Suppose that the pressure function $p$ satisfies (\ref{whitedwarf1}) with $\gamma=4/3$  and
   (\ref{origin}) holds.  Assume that the angular momentum (per unit mass) $J(m)=\sqrt {L(m)}$ satisfies (2.14), then (\ref{negative}) holds for $\r\in W_M$ for any $0<M<+\infty$.
      \end{thm}
 \begin{proof} By the existence theorem in \cite{FT2}, if (2.14) is satisfied, then for any $0<M<+\infty$, there exits $\tilde \r\in W_{M, S}$
 such that $F(\tilde \r)=\inf _{\r\in W_{M, S}} F(\r)$. Also, all the properties of $\tilde \rho$ in Theorem 2.1 are satisfied. Moreover, the regularity of the boundary $\pa G$ is smooth enough to apply the Gauss-Green formula (cf. \cite{CF}). The proof now follows exactly as in  Theorem 4.2.\end{proof}

We finally turn to the case of rotating supermassive stars.
\begin{thm}\label{thm4.4} Consider suppermasive star; i.e.,
\be\label{ppxx8} p(\r)=k\r^{4/3}, \qquad k>0{\rm ~is~a~constant}. \ee If there exists $\hat \r\in W_M$ such that $\hat \r\in C^1(G)\cap C(\RR^3)$ is a steady
state solution of the Euler-Poisson equation with the velocity field $\hat {\bf v}=(-\f{x_2 \sqrt L(m_{\hat \r}(r))}{r},  \f{x_1 \sqrt L(m_{\hat \r}(r))}{r}, 0)$  in an open bounded set $G\subset \RR^3$ with the Lipschitz boundary $\pa G$, i.e.,
\be\label{ppxx9}\begin{cases} &\nabla_x p(\hat \r)=\hat \r\nabla_x(B\hat \r)+\hat \r L(m_{\hat \r})r(x)^{-3}{\bf e}_r, \ x\in G,\\
& \hat \r=0, \qquad x\in \RR^3-G.\end{cases}
\ee
then (\ref{negative}) holds provided $L$ satisfies (2.9) and \be\label{ppxx10}
L(m_0)>0, {\rm~for~ some~} m_0\in (0, M).\ee\end{thm}
\begin{proof} Following along the same lines as (4.7)-(4.10), we obtain the same equality as (4.11). Therefore,
\be\label{ppxx10'} F(\hat \r)=-\f{1}{2}\int \hat \r(x) L(m_{\hat \r}(r(x))r^{-2}(x)dx,\ee
in view of (\ref{ppxx8}) and (4.11). Since $\hat \r\in C^1(G)\cap C(\RR^3)$ and   $\hat \r=0$ for $x\in \RR^3-G$, it is easy to show that
$m_{\hat\r}(r)$ is continuous in $r$. Moreover, $m_{\hat\r}(0)=0$ and $m_{\hat\r}(R)=M$, where $R=\max_{x\in \bar G}(r(x)$. Therefore, there
exits $r_0\in (0, M)$ such that
\be\label{ppxx11}
m_{\hat\r}(r_0)=m_0,\ee
where $m_0$ is the constant in (\ref{ppxx10}). Thus,
\be\label{ppxx11'}
L(m_{\hat\r}(r_0))>0,
\ee
in view of (\ref{ppxx10}). Since $m_{\hat\r}(r)$ is continuous in $r$ and $L(m)$ is continuous in $m$, we conclude that
\be\label{ppxx12} \int \hat \r(x) L(m_{\hat \r}(r(x))r^{-2}(x)dx>0.\ee
The inequality (\ref{negative}) now follows from
                                (\ref{ppxx10'})).
\end{proof}

The preceding theorems, together with Theorem 3.2 show that  polytropes ($p(\r)=k\r^{\g}$) with $\gamma> 4/3$  and  White Dwarf stars, in both the rotating and non-rotating cases, as well as rotating Supermassive stars are dynamically stable.  Moreover, if the angular momentum distribution is not everywhere positive and the pressure p behaves asymptotically near infinity like $\rho^{4/3}$, then dynamic stability holds only under a (Chandrasekhar) mass restriction, $M\le  M_c$.

\section{Nonlinear Dynamical Stability of Non-Rotating White Dwarf Stars With General Perturbations}
The  dynamical stability results in Section 3 apply for axi-symmetric perturbations. In this section,  we  prove the nonlinear dynamical stability for non-rotating white dwarf stars with general perturbations. For white dwarf stars, as mentioned before, the pressure function satisfies
 \be\label{55.1}
 p\in C^1[0, +\infty),\  \lim_{\rho\to 0+}\f{p(\r)}{\r^{\g_1}}=\beta,\  \lim_{\r\to \infty}\frac{p(\r)}{\r^{\gamma}}=K,\   p'(\r)>0  {~\rm for~} \r>0,
 \ee
 where $\g_1>4/3$, $0<\beta<+\infty$ and  $0<K<+\infty$ are constants. In this section, we always assume that the pressure function satisfies (\ref{55.1}).  First, we define for $0<M<+\infty$,
\begin{align}\label{5.2} X_M&=\{\r: \RR^3\to \RR, \rho\ge 0, a.e.,\
 \int\rho(x)dx=M,\notag\\
&\int
  [A(\rho(x))+\frac{1}{2}\rho(x) B\rho(x)]dx<+\infty \},  \end{align}
  where $A(\r)$ is the function given in (2.5).
 For $\rho\in X_M$,  we define the {\bf energy functional} $G$ for non-rotating stars by
  \begin{align}\label{55.3}
  G(\rho)=\int
  [A(\rho(x))-\frac{1}{2}\rho(x) B\rho(x)]dx.
  \end{align}
 We begin with the following theorem.
 \begin{thm}\label{5.1'} Suppose that the pressure function $p$ satisfies (\ref{55.1}).
Let $\tilde \r_N$ be a minimizer of the energy functional $G$ in $X_M$ and let
\be\label{G1'}
\Gamma_N=\{x\in \RR^3:\  \tilde \r_N(x)>0\}, \ee then
 there exists a constant $\lambda_N$
such that
\be\label{lambda111}
\begin{cases}
& A'(\tilde \r_N(x))-B\tilde \r_N(x)=\lambda_N, \qquad x\in \Gamma_N,\\
&-B\tilde \r_N(x)\ge \lambda_N, \qquad x\in \RR^3-\Gamma_N.\end{cases}\ee
\end{thm}
The proof of this theorem is well-known, cf. \cite{Rein} or \cite{AB}.
\begin{rem} 1) We call the minimizer $\tilde\rho_N$  of the functional $G$ in $X_M$  a non-rotating star solution.\\
2) It follows from \cite{liebyau} that the minimizer $\tilde\rho_N$  of the functional $G$ in $X_M$ is actually radial, and has a compact support.
\end{rem}
\vskip 0.5cm
Similar to Theorem 3.1, we have the following compactness theorem.
\begin{thm}\label{aa'}  Suppose  that  the pressure function $p$ satisfies (\ref{55.1}).  There exists a constant ${M}^c$ ($0<M^c<\infty)$  such that if $M<M^c$, then the following hold:\\
(1) \be\label{negative'}  \inf_{\r\in X_M}G(\r)<0,\ee
 (2) for $\r\in X_M$, \be\label{us11} \int A(\r)(x) dx\le C_1 G(\r)+C_2,\ee
for some positive constants $C_1$ and $C_2$,\\
(3)  if  $\{\r^i\}\subset X_M $ is a minimizing sequence for the functional $G$,
then there exist a sequence of  translations $\{x^i\}\subset \RR^3$,  a subsequence of $\{\r^i\}$,  (still labeled  $\{\r^i\}$),  and a function $\tilde \r_N\in X_M$, such that for any $\epsilon>0$ there exists $R>0$ with
\be\label{2.15''}
\int_{|x|\ge R}T\r^i(x)dx\le \epsilon, \quad i\in \mathbb{N},\ee
and
\be\label{2.16'} T\r^i(x)\rightharpoonup \tilde \r_N,\  weakly~in~L^{4/3}(\RR^3),\ as\  i\to \infty,\ee
where $T\r^i(x):=\r^i(x+x^i)$. \\
\noindent Moreover\\
(4)
\be\label{2.17'}\nabla B (T\rho^i)\to \nabla B(\tilde \r_N)~ strongly~ in~L^2(\RR^3),\ as\  i\to \infty, \ee
and\\
 (5) $\tilde \r$ is a minimizer of $G$ in $X_M$.
\end{thm}

\begin{proof} First, the proofs of (1) and (2) are the same as Theorems 4.1 and 4.2 by taking $L=0$ (it is easy to check the axial symmetry
is not used the the proof of Theorems 4.1 and 4.2 if $L=0$).  Lemmas 3.4, 3.5 and 3.7  still hold by taking $\gamma=4/3$ and $L=0$,
and replacing $W_M$ by $X_M$,  $F$ by $G$ and $f_M$ by $\inf_{\r\in X_M}G(\r)$.
Also, it is easy to check that (3.25)-(3.29) in the proof of Lemma 3.6 still hold by replacing $f_M$  by $\inf_{\r\in X_M}G(\r)$. Therefore,
following the proof of Lemma 3.6, we conclude:\\
If $\{\r^i\}\subset X_M $ is a minimizing sequence for $G$,  then there exists constant $\delta_0>0$,  $i_0\in \mathbf{N}$ and $x^i\in \RR^3$,  such that
$$\int_{B_1(x^i)}\r^i(x)dx\ge \delta_0, \ i\ge i_0.$$
Therefore, if we let
\be T\r^i(x):=\r^i(x+x^i), \ee
then
$$\int_{B_1(0)}T \r^i(x)dx\ge \delta_0, \ i\ge i_0.$$
This is similar to (3.39). Having established this inequality and the other analogues of Lemmas 3.4, 3.5 and 3.7, we can prove this theorem in a
similar manner as the proof of Theorem 3.1. \end{proof}
For the stability, we consider the Cauchy problem (1.1) with the initial data (3.53). We {\it do not} assume that the initial data have any symmetry. \\
Let $\tilde\r_N$ be a minimizer of $G$ on $X_M$ and  $\lambda_N$ be the constant in (\ref{lambda111}). For $\r\in X_M$, we define
 \begin{align} d(\r, \tilde \r_N)&=\int
\{[A(\rho)-A(\tilde \r_N)]
-(\r-\tilde \r_N)(\lambda_N+B\tilde \r_N\}dx,\notag\\
&=\int
\{[A(\rho)-A(\tilde \r_N)]
-B\tilde \r_N (\r-\tilde \r_N)\}dx,
\end{align}
where we have used the identity
$$\int\r dx=\int\tilde\r_Ndx=M, $$
for $\r\in X_M$.  By a similar argument as (3.82) and (3.83), we have
\be d(\r, \tilde \r_N)\ge 0, \ee
for any $\r\in X_M$,  in view of (\ref{lambda1'}).  Our nonlinear stability theorem of non-rotating white dwarf star solutions is the following theorem, which extends the results in \cite{Rein}.
\begin{thm}\label{th5.1'} Suppose that the pressure function satisfies (\ref{55.1}).   Let $\tilde \r_N$ be a minimizer of the functional $G$ in $X_M$, and assume that it  is unique up to a translation $\r_N(x)\to\r_N(x+y)$.    Let $(\r, {\bf v}, \Phi)(x, t)$ be a weak solution of the Cauchy problem (\ref{1.1}) and (\ref{initial}) satisfying $$\int \r(x, t)=\int \r_0(x)=\int \r_N(x)dx=M.$$   If the total energy $E(t)$ (cf. (\ref{energy})) is non-increasing with respect to $t$, then there exists a constant ${M}^c$ ($0<M^c<\infty$)  such that if $M<M^c$,
then for every $\epsilon>0$, there exists a number $\delta>0$ such that if
\be d(\r_0, \tilde \r_N)+\f{1}{8\pi}||\nabla B\r_0-\nabla B\tilde \r_N||_2^2
+\f{1}{2}\int \r_0(x)(|v_0|^2)(x)dx
<\delta,\ee
then  there is a translation  $y\in \RR^3$  such that, for every $t>0$
\be d(\r(t), T^y\tilde \r_N)+\f{1}{8\pi}||\nabla B\r(t)-\nabla BT^y\tilde \r_N||_2^2
+\f{1}{2}\int \r(x, t)|v(x, t)|^2)dx
<\epsilon,
\ee
where $T^y\tilde \r_N(x)=:\tilde \r_N(x+y).$
\end{thm}
The proof of this theorem follows from the compactness result (Theorem \ref{aa'}),  and the arguments in \cite{LS1} and \cite{Rein},  and is thus omitted.

\centerline{\bf Acknowledgments}
Luo  was supported in part by  the National Science Foundation under Grants  DMS-0606853 and DMS-0742834. Smoller  was supported in part by  the National Science Foundation under Grant  DMS-0603754.

\bibliographystyle{plain}

\begin{thebibliography}{99}
\bibitem{Ambrosio} L. Ambrosio, Transport equation and Cauchy problem for $BV$ vector fields. Invent. Math. 158, no. 2, 227--260 (2004).
\bibitem{AB}G. Auchmuty and R. Beals , Variational solutions
of some nonlinear free boundary problems, Arch. Rat.  Mech.
Anal.43, 255-271 (1971).

\bibitem{Au}G. Auchmuty, The global branching of rotating stars,  Arch. Rat. Mech.
Anal. 114, 179-194 (1991).

\bibitem{CF} L. Caffarelli and A. Friedman, The shape of axi-symmetric
rotating fluid, J. Funct. Anal.,  694, 109-142 (1980).

\bibitem{chan} S. Chandrasekhar,  Phil. Mag. 11, 592 (1931); Astrophys. J. 74, 81 (1931); Monthly Notices
Roy. Astron. Soc. 91, 456 (1931); Rev. Mod. Phys. 56, 137 (1984).


\bibitem{ch}S. Chandrasekhar, Introduction to the Stellar
Structure, {\it University of Chicago Press} (1939).

\bibitem{Li2} Chanillo, Sagun and Li, Yan Yan,
On diameters of uniformly rotating stars.
Comm. Math. Phys. 166,  no. 2, 417--430 (1994).



\bibitem{chenfrid2}G. Q.  Chen \& H. Frid, Extended divergence-measure fields and the Euler equations for gas dynamics, Comm. Math. Phys. 236, no. 2, 251-280 (2003).

\bibitem{lp}R.J. Di Perna \& P. L. Lions, Ordinary differential equations, transport theory and Sobolev spaces. Invent. Math. 98, 511--547 (1989).



\bibitem{LY} Y. Deng, T.P. Liu, T. Yang \& Z.Yao,  Solutions of Euler-Poisson
equations for gaseous stars,  Arch. Rat. Mech. Anal. 164 , no.
3, 261--285 (2002).

\bibitem{fowler} R. H. Fowler, Monthly Notices Roy. Astron. Soc. 87, 114 (1926).





\bibitem{FT1} A. Friedman \& B. Turkington, Asymptotic estimates for an axi-symmetric rotating fluid. J. Func. Anal. 37, 136-163 (1980).

\bibitem{FT2} A. Friedman \& B. Turkington, Existence and dimensions of a rotating white dwarf. J. Diff. Eqns. 42, 414-437 (1981).

\bibitem{Jang} J. Jang, Nonlinear instability in gravitational Euler-Poisson system for
$\gamma =\frac{6}{5}$, to appear in Arch. Rat. Mech. Anal..


\bibitem{guo1} Y. Guo \& G. Rein,  Stable steady states in stellar dynamics, Arch. Rat. Mech.
Anal. 147, 225-243 (1999).

\bibitem{guo2} Y. Guo \& G. Rein, Stable models of elliptical galaxies, Mon. Not. R.
Astron. Soc. 344, 1296-1306, (2002).

\bibitem{GT} D. Gilbarg \& N. Trudinger,  Elliptic Partial Differentail Equations
of Second Order (2nd ed.), {\it Springer} (1983).

\bibitem{Humi}M. Humi,  Steady states of self-gravitating incompressible fluid in two dimensions. J. Math. Phys. 47,  no. 9, 093101, 10 pp  (2006).

\bibitem{landau} L, Landau,  Phys. Z. Sowjetunion 1, 285 (1932)

\bibitem{lax} P. Lax, Shock Waves and Entropy, Contribution to Nonlinear Functional Analysis, E. A. Zarantonello, ed., {\it Academic Press,} New York 603-634 (1971).

\bibitem{lebovitz} N. R. Lebovitz, The virial tensor and its application to self-gravitating fluids, Astrophys. J. 134 (1961) 500--536.
\bibitem{lebovitz1} N. R. Lebovitz \& A. Lifschitz, Short-wavelength instabilities of Riemann ellipsoids,  Philos. Trans. Roy. Soc. London Ser. A 354, no. 1709, 927--950 (1996).

\bibitem{Li1} Li, Yan Yan,  On uniformly rotating stars,  Arch. Rat. Mech. Anal. 115, no. 4, 367--393 (1991).

\bibitem{liebyau} E. H. Lieb \& H. T.  Yau,  The Chandrasekhar theory of stellar collapse as the limit of quantum mechanics. Comm. Math. Phys. 112, no. 1, 147--174 (1987).


\bibitem{lin} S. S. Lin, Stability of gaseous stars in spherically symmetric motions,  SIAM J. Math. Anal. 28, no. 3, 539--569 (1997).

\bibitem{lions} P.L. Lions, The concentration-compactness principle in the calculus of variations. The locally compact case. Part I, Ann. Inst.
H. Poincar$\rm \acute{e}$ Anal. Non Lin$\rm \acute{e}$aire, 1, 109-145 (1984).

\bibitem{LS}T. Luo \& J. Smoller, Rotating fluids with self-gravitation in bounded domains. Arch. Rat. Mech. Anal. 173, no. 3, 345--377 (2004).

\bibitem{LS1}T. Luo \& J. Smoller,  Existence and Nonlinear Stability of Rotating Star Solutions of the Compressible Euler-Poisson Equations,
Preprint, http://arxiv.org/PS\_cache/gr-qc/pdf/0703/0703033v1.pdf, (to appear in Arch. Rat. Mech. Anal.).



\bibitem{MK} T. Makino, Blowing up of the Euler-Poisson equation
for the evolution of gaseous star, Transport Theory and
Statistical Physics, 21, 615-624 (1992).

\bibitem{RS}M. Reed \& B.Simon, Methods of Modern Mathematical Physics II: Fourier Analysis,
Self-Adjointness, {\it Academic Press}, NewYork (1975).

\bibitem{rein1} G. Rein, Reduction and a concentration-compactness principle for energy-casimir functionals,
SIAM J. Math. Anal. Vol 33, No. 4,  896-912 (2001).

\bibitem{Rein} G. Rein, Non-linear stability of gaseous stars. Arch. Rat. Mech. Anal. 168, no. 2, 115--130 (2003).


\bibitem{ST} S. H. Shapiro \& S. A. Teukolsky,  Black Holes, White Dwarfs, and Neutron Stars, {\it WILEY-VCH}, (2004)


\bibitem{smoller} J. Smoller, Shock Waves and Reaction-Diffusion Equations, 2nd Ed.), Springer, Berlin, New York (1994).


\bibitem{Ta} J. L. Tassoul, Theory of Rotating Stars, {\it Princeton University Press}, Princeton (1978).

\bibitem{Wang}D. Wang,  Global Solutions and Stability for Self-Gravitating
Isentropic Gases, J. of Math. Anal. \& Appl., 229, 530-542 (1999).

\bibitem{wein} S. Weinberg , Gravitation and Cosmology {\it John Wiley
and Sons}, New York , 1972.
 \end{thebibliography}

Tao Luo\\
 E-mail: tluo@wpi.edu\\or\\tl48@georgetown.edu

Joel Smoller\\
 
E-mail: smoller@umich.edu\\

\end{document}